\documentclass{llncs}
\usepackage[utf8]{inputenc}
\usepackage{amsmath}
\usepackage{xcolor}
\usepackage{amssymb}
\usepackage{mathrsfs} 
\usepackage{enumitem}
\usepackage{hyperref}
\usepackage{cleveref}
\usepackage{algorithm}
\usepackage{algpseudocode}
\usepackage{graphicx}

\usepackage{diagbox}

\usepackage{balance} 
\usepackage{lineno}

\date{Jan 21, 2024}

\begin{document}



\title{Equilibrium Analysis of Customer Attraction Games}
\author{
Xiaotie Deng  \inst{1} \and
Hangxin Gan \inst{2} \and
Ningyuan Li \inst{1} \and
Weian Li \inst{3} \and
Qi Qi \inst{4}
}

\institute{Center on Frontiers of Computing Studies, School of Computer Science, Peking University, Beijing, China\\
\email{\{xiaotie, liningyuan\}@pku.edu.cn}\and
School of Mathematical Sciences, Nankai University, Tianjin, China\\ \email{hangxin.gan@mail.nankai.edu.cn}\and
School of Software, Shandong University, Jinan, China\\
\email{weian.li@sdu.edu.cn}
\and
Gaoling School of Artificial Intelligence, Renmin University of China, Beijing, China\\
\email{qi.qi@ruc.edu.cn}
}
\maketitle

\begin{abstract}
  We introduce a game model called ``customer attraction game'' to demonstrate the competition among online content providers. In this model, customers exhibit interest in various topics. Each content provider selects one topic and benefits from the attracted customers. We investigate both symmetric and asymmetric settings involving agents and customers. In the symmetric setting, the existence of pure Nash equilibrium (PNE) is guaranteed, but finding a PNE is PLS-complete. To address this, we propose a fully polynomial time approximation scheme to identify an approximate PNE. Moreover, the tight Price of Anarchy (PoA) is established. In the asymmetric setting, we show the nonexistence of PNE in certain instances and establish that determining its existence is NP-hard. Nevertheless, we prove the existence of an approximate PNE. Additionally, when agents select topics sequentially, we demonstrate that finding a subgame-perfect equilibrium is PSPACE-hard. Furthermore, we present the sequential PoA for the two-agent setting.
\end{abstract}







\section{Introduction}
\label{Sec:Intro} 
The widespread adoption of the Internet has prompted an increasing number of companies to shift their focus towards online marketplaces. The global digital content market, for instance, has expanded significantly, reaching a staggering 169.2 billion in 2022 and estimated to surpass 173.2 billion in 2023\footnote{\url{https://www.marketresearchfuture.com/reports/digital-content-market-11516}}. Concurrently, the rise of social media platforms has resulted in a surge of individual content creators and users engaging in these platforms. YouTube, as of June 2023, boasts 2.68 billion active users, while TikTok accumulates 1.67 billion active users\footnote{\url{https://www.demandsage.com/youtube-stats/},\\ \url{https://www.demandsage.com/tiktok-user-statistics/}}. Given the vast number of users with diverse needs and interests, both companies and content creators must adopt suitable marketing strategies and select relevant content topics to attract their target customers effectively. As video producers aim to generate more traffic and subscribers, they often tailor their content topics based on the preferences of their target audience. However, with users exhibiting multiple interests and limited internet usage time, producers face stiff competition not only from other producers sharing the same topic but also from any producer targeting overlapping users. Similarly, on
digital advertising platforms, online business owners try to choose the right keywords or tags to attract a specific group of potential customers, while facing competition from other advertisers.

Upon examining the present competitive landscape of online customer acquisition, two noteworthy phenomena emerge. On the one hand, in the red ocean market (representing popular topics), excessive competition arises for a limited customer base, resulting in a waste of social resources. On the other hand, in the blue ocean market (representing unpopular topics), existing platforms fail to satisfy the demand of certain customers, leading to a loss of platform users. Inspired by these phenomena, to explore the underlying reasons and investigate the social welfare loss caused by competition, it is necessary to build a behavioral model for the customer attraction scenarios.

Traditionally, the problem of attracting customers is modeled as a location game in the offline markets. Retailers open physical stores with the aim of drawing nearby residents. The store location is determined by considering the number of potential customers within the service range and the competition from neighboring retailers. However,  attracting customers online differ significantly. In a location game, the competition typically occurs in a two-dimensional plane or a one-dimensional interval, primarily relying on distance to attract customers. Conversely, attracting online customers involves significantly more complexity as their preferences correspond to a high-dimensional latent feature vector. Consequently, the outcomes of location games are not directly transferrable to online scenarios.

In this work, we present a game model, called ``customer attraction game'', which captures the common features of online scenarios.  In this model, each customer expresses interest in specific topics, while each content producer (acting as an agent) chooses a topic and receives utility based on the customers they attract. However, if a customer is attracted by multiple agents, her utility is divided proportionally to reflect each agent's competitiveness. Particularly, when agents are symmetric, customers randomly select an agent to whom they are attracted. Moreover, when agents are asymmetric, the selection probability aligns with each agent's weight.

We delve into the simultaneous and sequential topic selection behaviors of agents, observing scenarios where agents select topics simultaneously or in a predetermined order.  The simultaneous model can capture that the film studios conceive the themes of their entries for a film festival, or advertisers who can change the targeted audience at a relatively low cost. Conversely, the sequential model applies to content creators on social media platforms who are cautious about changing topics to avoid losing followers. To understand the strategic behavior of competing agents, we focus on the stable state and raise the following questions: Does such a stable state exist? Can this stable state be reached efficiently? What is the performance of the stable state?

\subsection{Main Results}
Our model is representative and has wide application. Different from previous papers of location games, where customers usually live in one or two dimensions, our model is highly abstract and not restricted by dimensions, which can be regarded as a general version of traditional location games. This means that our results can be applied to other more abstract scenarios beyond facility location or political election contexts, opening up possibilities for broader applications.

(1) In a static symmetric game where all agents have equal weight and strategy space, we establish the existence of Pure Nash Equilibrium (PNE) and prove that finding a PNE is PLS-complete. To address this intractability, we propose a Fully Polynomial Time Approximation Scheme (FPTAS) that computes an $\epsilon$-approximate PNE. Additionally, we provide a tight Price of Anarchy (PoA) for the symmetric setting.

(2) In static asymmetric games, we generalize the results obtained in the symmetric setting when only the strategy spaces differ among agents. However, when the weights assigned to agents are different, we demonstrate the nonexistence of PNE and establish the NP-hardness of determining PNE existence. Specifically, we technically construct an counter-instance and prove that it does not admit a (2-$\epsilon$)-approximate PNE. Nonetheless, we show that there always exist a $O(\log w_{\textnormal{max}})$-approximate PNE which can also achieve the $O(\log W)$-approximately optimal social welfare, where $w_{\textnormal{max}}$ and $W$ are the maximum and sum of all agents' weights. 


(3) For sequential games where agents make decisions in order, we reveal the PSPACE-hardness of finding a subgame-perfect equilibrium (SPE), even in symmetric settings. Additionally, we analyze the sequential Price of Anarchy (sPoA) and show that it equals 3/2 for the two-agent case. Furthermore, we establish a lower bound of sPoA that approximately approaches to 2 for the general $n$-agent case.


We briefly introduce our technical highlights as follows:

(1) For the static symmetric game, we demonstrate the PLS-completeness of finding a PNE through a reduction from the local max-cut problem. Notably, this result remains applicable even for the simplest case where all agents have equal weights and strategy spaces. Additionally, since the total number of nodes is polynomial in the instance of our symmetric game, but the maximum edge weight in a local max-cut instance should be exponentially large, we technically encode exponentially large edge weights with polynomial number of nodes by utilizing the proportional allocation rule.


(2) For the static asymmetric game, we construct a concise example to show the nonexistence of PNE in the asymmetric game, and utilize it as a gadget in our reduction to prove the NP-hardness to determine the existence of PNE. Moreover, we introduce a technical lemma that guarantees the PNE equivalence between the asymmetric case and the case with symmetric strategy spaces, and establishes the non-existence of PNE and NP-hardness results under a nearly-symmetric case where only one agent has a different weight. 


\subsection{Related Work}

In terms of our model, the work of Goemans et al. \cite{GLMT04} and Bil{\`o} et al. \cite{BGM23} are the most relevant. Goemans et al. \cite{GLMT04} introduce a market sharing game that bears similarities to our symmetric static game. However, their focus is primarily on unweighted scenarios, without considering the weighted case. Additionally, they do not investigate the complexity of finding an equilibrium. In the study conducted by Bil{\`o} et al. \cite{BGM23}, a project game is introduced where participants select preferred projects to participate in. Similar to our asymmetric static model, projects are asymmetric and players gain rewards based on different weights. However, an important distinction is that our asymmetric static game allows players to choose multiple projects, whereas in their model, players can only select one project.

In addition, two concepts in recent papers are similar to the description of customers' preferences in our model: attraction interval of agents and tolerance interval of customers. With respect to the former concept, Feldman et al. \cite{FFO16} introduce the idea of limited attraction. They demonstrate that equilibria always exist and analyze the PoA and Price of Stability (PoS). Shen and Wang \cite{SW17} generalize this model to the case of any customer distribution.
For the latter concept, Ben-Porat and Tennenholtz \cite{BT17} investigate customers' choice rules based on a specific probability function. Cohen and Peleg \cite{CP19} focus on a model where customers have a tolerance interval and only visit shops within this range.
In comparison to these models, the attraction ranges in our model reflect both features simultaneously and can be viewed as a highly abstract version of attraction interval and tolerance interval.


Our paper is mainly related to two topics: location games and congestion games. The offline location game is commonly modeled by the classic Hotelling-Downs model \cite{H29,D57}. Some comprehensive surveys on this topic have been published in \cite{ELT93,B10,E11}. Additionally, the Hotelling-Downs model has a wide range of applications in the field of political elections, e.g., \cite{SS08,BC15,OEPR15,FGLLM16,HLST21,DEG22}. In recent decades, location problems have also been studied from the perspective of mechanism design (see \cite{PT13,FFG16,ACLLW20,CFLLW21,FP21}).


For the literature focusing on discrete location spaces, discrete customer location space is first introduced by \cite{S61}. Huang \cite{H11} examines the mixed strategies of three players with inelastic utility functions. N{\'u}{\~n}ez and Scarsini \cite{NS16} consider the discrete locations of agents and customers, where customers choose a shop based on their preferences.
Iimura and von Mouche \cite{IM21} focus on the general case of non-increasing utility functions with two players.
Recently, two-stage facility location games with strategic agents and clients have been studied \cite{KLM21,KLS23}, where both agents and clients are strategic and agents' utilities depend on the clients' equilibrium.
Unlike the utility functions in most previous papers, which were related to the distance between customers and agents, we consider customers being allocated based on the weights of agents.


The sequential version of the location game, known as the Voronoi game, is introduced by \cite{ACCG04}. In this game, two players compete for customers in a specific area by alternately locating points. The winning strategy for the second player is given. A series of papers have studied the Voronoi game on different graphs, such as cycles \cite{MMP08} and general networks \cite{DT07,BBDS15}.


In the field of algorithmic game theory, techniques used in location games can be related to the topic of congestion games, which is first introduced by \cite{R73}.
Monderer and Shapley \cite{MS96} prove that any finite potential game is equivalent to a congestion game. Gairing \cite{G09} introduces a variant of congestion games called covering games. It shows that in a covering game with a specially designed utility sharing function, every PNE can approximate maximum covering with a constant factor. The covering game corresponds to a special case of our model where agents are unweighted, but it does not consider weighted agents. Furthermore, Gairing \cite{G09} focuses more on the problem of mechanism design, while we mainly investigate equilibrium analysis and evaluation under a proportional allocation rule.
In addition, the fruitful research results of congestion games have contributed numerous techniques to classic concepts in game theory, such as the existence of PNE (\cite{ADKTTR08,HK12,KR15,GKK20}), PoA or PoS (\cite{RT02,CV07,CN09,CR09,ADGM11}), and sequential congestion games (\cite{CK05,LST12,JU14,CGMMP20}).


\subsection{Roadmap}
In Section \ref{Sec:Pre}, we formally introduce our customer attraction game. In Section \ref{Sec:basic}, we mainly focus on the symmetric setting and show a series of results about PNE. In Section \ref{Sec:asymmetric}, the asymmetric setting is investigated. We discuss the sequential version of our game in Section \ref{Sec:sequential}. In Section \ref{Sec:conclusion}, we give a summary of the whole paper and propose the directions for future work. 

\section{Model and Preliminaries}
\label{Sec:Pre}

In this section, we formally introduce our customer attraction game (CAG) (see Figure \ref{fig:pic1}). 
Assume that there are $n$ types of customers represented by $n$ discrete nodes\footnote{In the following, we sometimes use ``node'' to represent customers in language.} and define the set of nodes as $N = \{1,2, \cdots, n\}$. For each $j \in N$, its value $v_j$ denotes the number or importance of customers represented by node $j$. Without loss of generality, we assume that $v_j\in\mathbb{Z}_{\geq 1}$ for any $j\in N$.

There are $L$ candidate topics for agents. For each topic $l$, its attraction range 
is defined as $s_l\subseteq N$, 
i.e., the set of customers who are interested in this topic.
Define $\mathcal{S}= \{s_1, s_2, \cdots, s_L\}$ as the collection of attraction ranges of all candidate topics. We assume that $\bigcup_{l=1}^L s_l = N$, that is, each customer is attracted by at least one topic.

There are $m$ agents 
and the set of agents is denoted by $A=\{1,\cdots,m\}$. Each agent $i\in A$ has a weight $w_i$ and we suppose that $w_i\in\mathbb{Z}_{\geq 1}$. The strategy space of agent $i$ is defined as $\mathcal{S}_i\subseteq\mathcal{S}$ representing the available topics provided for agent $i$. 
Let the joint strategy space of all agents be $\vec{\mathcal{S}}=\times_{i=1}^m\mathcal{S}_i$. When each agent chooses a pure strategy $S_i\in\mathcal{S}_i$, a pure strategy profile is denoted by $\vec{S}=(S_1,\cdots,S_m)\in\vec{\mathcal{S}}$.

\begin{figure}[ht]
    \centering
    \includegraphics[width=0.6\textwidth]{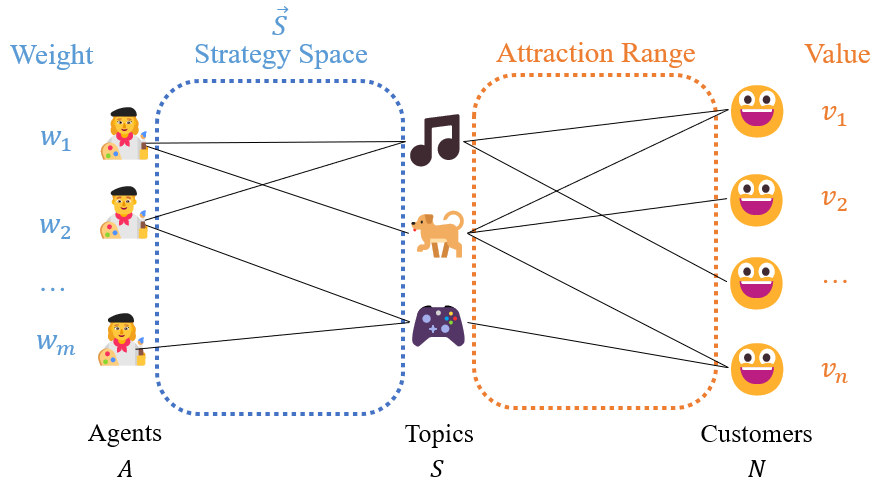}
    \caption{ 
    An example of customer attraction game. 
    Each agent is linked to all her available topics. A line from a topic to a customer means this customer is in the attraction range of this topic.}
    \label{fig:pic1}
\end{figure}

Given a pure strategy profile $\vec{S}$, we say that agent $i$ attracts the node $j$ if and only if $j\in S_i$, that is, customers on node $j$ are interested in the topic which agent $i$ selects. For each node $j\in N$, define the load function $c(j,\vec{S})$ as the total weight of the agents attracting $j$, formally,
$$c(j,\vec{S})=\sum_{i\in A:j\in S_i} w_i.$$
When node $j$ is attracted by at least one agent, each customer on it selects an agent attracting node $j$ with the probability proportional to each agent's weight. 
Therefore, in expectation, node $j$ distributes the $v_j$ customers to all agents attracting $j$ in proportion to their weights. The utility of each agent $i\in A$ is defined as the expected number of attracted customers, that is,
$$U_i(\vec{S})=\sum_{j\in S_i}\frac{w_i}{c(j,\vec{S})}v_j.$$

Now we formally define an instance of the customer attraction game.
\begin{definition}\label{def:ADHDG}
An instance of customer attraction game is defined by a tuple $\mathcal{I}=(N,$ $A,\vec{\mathcal{S}},\vec{w},\vec{v})$, where
\begin{itemize}
    \item $N$ is the set of all nodes;
    \item $A$ is the set of all agents; 
    \item $\vec{\mathcal{S}}=\times_{i=1}^m\mathcal{S}_i$, where $\mathcal{S}_i$ is the strategy space of agent $i$;
    \item $\vec{w}=(w_1, w_2, \cdots, w_m) \in \mathbb{Z}_{\geq1}^m$: the weights of all agents;
    \item $\vec{v}=(v_1, v_2, \cdots, v_n) \in \mathbb{Z}_{\geq1}^n$: the values of all nodes. 
\end{itemize}
\end{definition}

We also consider a special case called the symmetric CAG, where all agents are symmetric and all nodes are unit-value. Formally, 
\begin{itemize}
\item 
All agents share the same strategy space, that is, for all $i\in A$, $\mathcal{S}_i=\mathcal{S}$.
\item 
All agents have unit weight, that is, for all $i\in A$, $w_i=1$.
\item 
All nodes have unit value, that is, for all $j\in N$, $v_j=1$.
\end{itemize}
To distinguish, we call the general model defined in \Cref{def:ADHDG} as the asymmetric CAG. Sometimes we also discuss the settings in which only part of the above three components are restricted to be symmetric. For convenience, we denote such settings by listing the asymmetric components in the model as the prefix. For example, the notion ($\vec{\mathcal{S}},\vec{v}$)-asymmetric CAG means a CAG where the strategy spaces $\vec{\mathcal{S}}$ and the node values $\vec{v}$ can be different, while the agents' weights $\vec{w}$ are restricted to be the same.

Next we give the formal definition of pure Nash equilibria (PNE) in the CAG.

\begin{definition}
Given an instance $\mathcal{I}$ of customer attraction game, a strategy profile $\vec{S}$ is a pure Nash equilibrium if, for any $i\in A$ and $S'_i \in \mathcal{S}_{i}$, 
$$U_i(S_i,\vec{S}_{-i}) \geq U_i(S'_i,\vec{S}_{-i}).$$
We define PNE$(\mathcal{I})$ as the set containing all pure Nash equilibria of instance $(\mathcal{I})$.
\end{definition}  

Sometimes, PNE may not exist for some instances of CAG. For these instances, we focus on the approximate pure Nash equilibrium which is defined formally as

\begin{definition}\label{def:APNE}
    Given an instance of customer attraction game $\mathcal{I}=(N,A,\vec{\mathcal{S}},\vec{w},\vec{v})$, for any $\alpha \geq 1$, a strategy profile $\vec{S}$ is an $\alpha$-approximate pure Nash equilibrium if, for any $i\in A$ and $S'_i \in \mathcal{S}_i$,
    $$\alpha U_i(\vec{S}) \geq U_i(S'_i,\vec{S}_{-i}).$$
\end{definition}

In this model, for any instance of game $\mathcal{I}$,  we define the social welfare with respect to strategy profile, $\vec{S}$, as the sum of the utilities of all $m$ agents, which is also equal to the number of attracted customers under profile $\vec{S}$,
\begin{align*}
\text{SW}(\vec{S})= \sum_{i=1}^m U_i(\vec{S})= \sum_{\{j:\exists i, \text{such that} j\in S_i | \vec{S}\}} v_j.
\end{align*}
Let $\vec{S}^{*}$ be the strategy profile that achieves the optimal social welfare. We define the price of anarchy (PoA) (\cite{KP99}) for PNE of an instance as 
$$\text{PoA}(\mathcal{I})=\max_{\vec{S}^{NE}\in \text{PNE}(\mathcal{I})}\frac{\text{SW}(\vec{S}^{*})}{\text{SW}(\vec{S}^{NE})}.$$
Consequently, the PoA for CAG among all instances is defined as 
$$\text{PoA}=\sup_{\mathcal{I}} \max_{\vec{S}^{NE}\in \text{PNE}(\mathcal{I})}\frac{\text{SW}(\vec{S}^{*})}{\text{SW}(\vec{S}^{NE})}.$$

\section{Warm-up: The Symmetric Static Game}
\label{Sec:basic}

As a warm-up, we focus on the symmetric CAG, where all agents are symmetric and all nodes have unit value, as described in Section \ref{Sec:Pre} previously. This section aims to establish some fundamental insights about static CAG, as symmetric CAG falls under the broader category of asymmetric CAG (discussed in Section \ref{Sec:asymmetric}). Since we find that symmetric CAG is an exact potential game \cite{MS96}, some of the results in this section can be derived using the formal techniques of potential games. However, our main result in this section is the PLS-completeness of computing a PNE, which requires a different approach. This section is organized as: we first prove the existence of PNE by showing that any instance of symmetric CAG is an exact potential game. However, we demonstrate that finding a PNE is PLS-complete. Thus, we propose a fully polynomial-time approximation scheme (FPTAS) to output a ($1+\epsilon$)-approximate PNE. Finally, we give a tight PoA for the symmetric CAG.

Since all agents have identical weights of $1$ in the symmetric setting, the load function of one node degenerates to the number of agents attracting this node. That is,
$$c(j,\vec{S})=\left|\{i\in A:j\in S_i\}\right|$$ for any $j\in N$ and any $\vec{S}\in\vec{\mathcal{S}}$. To show the existence of PNE, we can construct a Rosenthal's potential function \cite{R73}, 
$$\Phi(\vec{S})=\sum_{j\in N}\sum_{k=1}^{c(j, \vec{S})} \frac{1}{k},$$
for any $\vec{S}\in\vec{\mathcal{S}}$ and show that any instance $\mathcal{I}$ is an exact potential game. In an exact potential game, whenever an agent changes its strategy to get a better utility, the potential function will also be increased. It means that any PNE is corresponding to a local maximum of $\Phi(\vec{S})$. Since the joint strategy space is finite, there must exist a joint strategy profile such that the potential function cannot be improved by changing the strategy of any agent, which guarantees the existence of PNE.

\begin{theorem}
\label{thm:PNE-existence-by-potential}
    For any instance $\mathcal{I}$ of symmetric customer attraction game, $\mathcal{I}$ is an exact potential game with respect to the potential function $\Phi(\vec{S})$ and the pure Nash equilibrium always exists.
\end{theorem}
\begin{proof}
We first show that $\mathcal{I}$ is an exact potential game, that is, for any strategy profile $\vec{S}\in\vec{\mathcal{S}}$ and any agent $i\in A$, when agent $i$ unilaterally deviates from $S_i$ to any strategy $S'_i\in\mathcal{S}_i$, the following equation holds, 
\begin{equation}U_i(S'_i,\vec{S}_{-i})-U_i(\vec{S})=\Phi(S'_i,\vec{S}_{-i})-\Phi(\vec{S}). \label{eq:potentialgame0}\end{equation}

Let $\vec{S'}=(S'_i,\vec{S}_{-i})$. By definition of $\Phi$, we have \begin{align}
\Phi(\vec{S'})-\Phi(\vec{S})&=\sum_{j\in N}\sum_{k=1}^{c(j,\vec{S'})} \frac{1}{k}-\sum_{j\in N}\sum_{k=1}^{c(j, \vec{S})} \frac{1}{k} 
=\sum_{j\in N}\left(\sum_{k=1}^{c(j,\vec{S'})} \frac{1}{k}-\sum_{k=1}^{c(j, \vec{S})} \frac{1}{k}\right). \label{eq:potentialgame1}
\end{align}
Observe that $\vec{S'}$ and $\vec{S}$ only differ in the strategy of agent $i$. Therefore, each $j\in N$ fits into one of the following three cases: (a) If $j\in S'_i\setminus S_i$, then $c(j,\vec{S'})=c(j,\vec{S})+1$; (b) If $j\in S_i\setminus S'_i$, then $c(j,\vec{S'})=c(j,\vec{S})-1$; (c) If $j\in (S'_i\cap S_i)\cup(N\setminus(S'_i\cup S_i))$, then $c(j,\vec{S'})=c(j,\vec{S})$. Thus, we have
\begin{align}
\sum_{j\in N}\left(\sum_{k=1}^{c(j,\vec{S'})} \frac{1}{k}-\sum_{k=1}^{c(j, \vec{S})} \frac{1}{k}\right)&=\sum_{j\in S'_i\setminus S_i}\frac{1}{c(j,\vec{S'})}-\sum_{j\in S_i\setminus S'_i}\frac{1}{c(j,\vec{S})} \nonumber\\
&=\sum_{j\in S'_i}\frac{1}{c(j,\vec{S'})}-\sum_{j\in S_i}\frac{1}{c(j,\vec{S})} 
=U_i(\vec{S'})-U_i(\vec{S}). \label{eq:potentialgame2}
\end{align}
The second line holds because $\sum_{j\in S'_i\cap S_i}\frac{1}{c(j,\vec{S'})}=\sum_{j\in S'_i\cap S_i}\frac{1}{c(j,\vec{S})}$.
Combining \Cref{eq:potentialgame1} and \Cref{eq:potentialgame2}, we immediately get \Cref{eq:potentialgame0}, which shows that $\mathcal{I}$ is an exact potential game with potential function $\Phi$. 

Define the neighborhood of $\vec{S}$ as all strategy profiles obtained by changing the strategy of at most one agent in $\vec{S}$, 
$$\mathcal{N}(\vec{S})=\{(S'_i,\vec{S}_{-i})\in\vec{\mathcal{S}}:S'_i\in\mathcal{S}_i,i\in A\}.$$ 
By definition, we know that $\vec{S}$ is a PNE if and only if $U_i(S'_i,\vec{S}_{-i})-U_i(\vec{S})\leq 0$ holds for any $i\in A$ and $S'_i\in \mathcal{S}_i$, which is equivalent to that $\Phi(\vec{S}')-\Phi(\vec{S})\leq 0$ holds for any $\vec{S}'\in\mathcal{N}(\vec{S})$ and $\vec{S}$ is a local maximum. Since $|\vec{\mathcal{S}}|$ is finite, there exists an $\vec{S}\in\vec{\mathcal{S}}$ such that $\Phi(\vec{S})=\max_{\vec{S}'\in\vec{\mathcal{S}}}\Phi(\vec{S}')$, which means that  $\Phi(\vec{S})\geq\Phi(\vec{S}')$ for any $\vec{S}'\in\mathcal{N}(\vec{S})$. It means that $\vec{S}$ is a PNE.
\end{proof}

Note that the proof of \Cref{thm:PNE-existence-by-potential} only requires the symmetry of agents' weights. Generally, for any instance $\mathcal{I}$ of ($\vec{\mathcal{S}},\vec{v}$)-asymmetric CAG, we can similarly extend the definition of $\Phi$ to  $$\Phi(\vec{S})=\sum_{j\in N}v_j\sum_{k=1}^{c(j,\vec{S})} \frac{1}{k},$$ and one can check that $\Phi(\vec{S})$ is still a potential function.
\begin{corollary}
\label{coro:PNE-existence-by-potential}
    For any instance $\mathcal{I}$ of ($\vec{\mathcal{S}},\vec{v}$)-asymmetric customer attraction game, the pure Nash equilibrium always exists.
\end{corollary}

Knowing the existence of PNE, a natural way to find a PNE is the best-response dynamics, roughly speaking, which starts with some strategy profile, and iteratively picks up one agent 
and let this agent deviate to the most beneficial strategy. By the existence of the potential function, it guarantees that the best-response dynamics can find a PNE in finite steps. However, this may take exponential number of steps. We prove that, finding a PNE in the symmetric CAG is PLS-complete. Note that this is a strong hardness result (analogous to the concept of strong NP-hardness), since it do not require the node values to be exponentially large, and the result even holds for the simplest CAG where all agents have the same weight and strategy space, and all nodes have unit value. 
\begin{theorem}\label{thm:PLS-complete}
Finding a Pure Nash Equilibrium of symmetric CAG is PLS-complete.
\end{theorem}
\begin{proof}
We have shown that finding a PNE of symmetric CAG (and ($\vec{\mathcal{S}},\vec{v}$)-asymmetric CAG) is equivalent to finding a local maximum of the potential function $\Phi(\vec{S})$. Therefore, the problem of searching a PNE of symmetric CAG is in the class PLS. To prove the PLS-completeness, we first give a reduction from the local max-CUT problem \cite{NRTV07} to finding a PNE in ($\vec{\mathcal{S}}$)-asymmetric CAG, as stated in Lemma \ref{lem:PLS-ADHDG}. Then, we reduce searching the PNE of any ($\vec{\mathcal{S}}$)-asymmetric CAG to searching the PNE of a symmetric CAG by Lemma \ref{lem:PLS-unionizing-strategy}. 

To prove the PLS-hardness of computing a PNE, a common idea is to reduce from the local max-cut problem, as adopted by \cite{G09} to establish the PLS-completeness of equilibrium computation in the covering games. However, in the symmetric CAG, a crucial difference is that each node is unit-valued, representing a single customer. It is important to note that the total number of nodes in a CAG instance is polynomial, whereas the PLS-hardness of the local max-cut problem requires the maximum edge weight in an instance to be exponentially large. Otherwise, a local max-cut can be found in polynomial time by standard local search. Therefore, it is unfeasible to express the edge weights directly by the number of nodes. Instead, the main focus of our proof lies in utilizing the proportional allocation rule of node values to encode exponentially large edge weights, using only polynomial number of nodes. Technically, we show that one can proportionally decrease the exponentially large integers into rational numbers in $[0,1]$, and express each number as the sum or difference of a series of fractions, such that all numerators and denominators, as well as the number of fractions, are polynomially bounded. This fact enables the construction of an edge gadget, which is employed to complete the reduction.



\begin{lemma}\label{lem:PLS-ADHDG}
Finding a Pure Nash Equilibrium of ($\vec{\mathcal{S}}$)-asymmetric CAG is PLS-hard.
\end{lemma}
\begin{proof}
We reduce the local max-cut problem to finding a PNE of an instance of ($\vec{\mathcal{S}}$)-asymmetric CAG.

Given an edge-weighted undirected graph $G(V,E,W)$, where the weight of each edge $e\in E$ is $W(e)\in\mathbb{Z}_{\geq1}$, a cut $(V',V'')$ is a partition of the vertex set $V$, such that $V'\cup V''=V,V'\cap V''=\emptyset$, and the objective function is the total weights of edges in a cut, defined as $\sum_{e=(u_1,u_2)\in E}W(e)\mathrm{I}[(u_1\in V'\land u_2\in V'')\vee(u_1\in V''\land u_2\in V')]$. The local max-cut problem asks to find a cut $(V',V'')$ such that moving any single vertex from one side to another does not increase the weight of the cut.

Denote $n_V=|V|$ and assume that the vertices are labeled as $V=\{u_1,\cdots,u_{n_V}\}$.
For convenience, any partition can be encoded as a vector $\vec{x}\in\{1,-1\}^{n_V}$, such that $V'(\vec{x})=\{u_i\in V:x_i=1\}$ and $V''(\vec{x})=\{u_i\in V:x_i=-1\}$. Let $CUTWEIGHT(\vec{x})$ denote the weight of the cut $(V'(\vec{x}),V''(\vec{x}))$, then we have $CUTWEIGHT(\vec{x})=\sum_{e=(u_{i_1},u_{i_2})\in E}\frac{1-x_{i_1}x_{i_2}}{2}W(e)$.

The total value of nodes in an ($\vec{\mathcal{S}}$)-asymmetric CAG is $|N|$, which cannot be exponentially large. This causes some technical difficulty to represent the edge weights in $G$. We will build a gadget consisting of a group of nodes and dummy players for each edge to encode its weight.

We construct an instance $\mathcal{I}=(N,A,\vec{\mathcal{S}},\vec{w},\vec{v})$ of ($\vec{\mathcal{S}}$)-asymmetric CAG, where
\begin{itemize}
    \item The node set $N$ consists of $|E|$ groups of nodes. Specifically, for each $e\in E$, we construct a group of nodes, called the edge gadget $GAD_e$.
    \item Let $A=A_1\cup A_2$, where $A_1=\{1,\cdots,n_V\}$. Each agent $i\in A_1$ corresponds with the vertex $u_i\in V$. Each agent $i'\in A_2$ is a dummy agent forced to attract a single node $POS(i')\in N$, by setting the strategy space of $i'$ as $\mathcal{S}_{i'}=\{\{POS(i')\}\}$. $A_2$ consists of $|E|$ groups of dummy players, each associated with an edge gadget.
    \item Let $\mathcal{S}_i=\{s_{i,1},s_{i,-1}\}$ for all $i\in A_1$. The elements in each $s_{i,1}$ and $s_{i,-1}$ will be determined later. Note that we can define a one-to-one correspondence between $\vec{\mathcal{S}}$ and $\{-1,1\}^{n_V}$, denoted as $\vec{S}(\vec{x})=(S_i(\vec{x}))_{i\in A}$, such that $S_i(\vec{x})=s_{i,x_i}$ for any $i\in A_1$, and $S_{i'}(\vec{x})=\{POS(i')\}$ for any $i'\in A_2$.
    \item $\vec{w}\equiv 1$ and $\vec{v}\equiv 1$ since they are symmetric.
\end{itemize}
Now we give the construction of the edge gadget $GAD_e$ for each $e=(u_{i_1},u_{i_2})\in E$. To build $GAD_e$, we add a group of nodes to $N$, which may be attracted by agents $i_1$ and $i_2$, and we add a group of dummy players to $A_2$.

Intuitively, suppose we build two nodes $r$, $q$ with $v_r=v_q=1$ and add them to be attracted as that $r$ is in $s_{i_1,1}$ and $s_{i_2,1}$, $q$ is in $s_{i_1,-1}$ and $s_{i_2,-1}$. Then when $x_{i_1}\neq x_{i_2}$, each of $r$ and $q$ is attracted by exactly one of $i_1$ and $i_2$; when $x_{i_1}=x_{i_2}$, one of $r$ and $q$ is attracted by both $i_1$ and $i_2$, while the other one is attracted by neither. In the former case, $r$ and $q$ make a total contribution of $1+1=2$ to the potential function. In the latter case, $r$ and $q$ make a total contribution of $\frac32+0=\frac32$ to the potential function. This difference can be utilized to encode the edge weight $W(e)$. Note that since $W(e)$ can be exponentially large, we cannot not feasible to simply repeat this structure for $W(e)$ times to get a total difference proportional to $W(e)$. To improve the flexibility of the construction, suppose that we add $d$ dummy agents with unit weights on both $r$ and $q$, then the difference between the cases that $x_{i_1}\neq x_{i_2}$ and $x_{i_1}=x_{i_2}$ in the contribution to the potential function becomes $\frac1{d+1}-\frac1{d+2}$. This structure helps us to encode rational numbers.

$GAD_e$ is specified by some non-negative integers $\{d^{e+}_k\}_{k=1,\cdots,L^{e+}}$ and $\{d^{e-}_k\}_{k=1,\cdots,L^{e-}}$, such that:
\begin{itemize}
    \item For each $k=1,\cdots,L^{e+}$, we create two nodes $r^{e+,k}_{0}$ and $r^{e+,k}_{1}$, and then we create $d^{e+}_k$ dummy agents forced to attract $r^{e+,k}_{0}$, and create $d^{e+}_k$ dummy agents forced to attract $r^{e+,k}_{1}$. We add $r^{e+,k}_{0}$ into $s_{i_1,1}$ and $s_{i_2,1}$, add $r^{e+,k}_{1}$ into $s_{i_1,-1}$ and $s_{i_2,-1}$.
    \item For each $k=1,\cdots,L^{e-}$, we create two nodes $r^{e-,k}_{0}$ and $r^{e-,k}_{1}$, and then we create $d^{e-}_k$ dummy agents forced to attract $r^{e-,k}_{0}$, and create $d^{e-}_k$ dummy agents forced to attract $r^{e-,k}_{1}$. We add $r^{e-,k}_{0}$ into $s_{i_1,1}$ and $s_{i_2,-1}$, add $r^{e-,k}_{1}$ into $s_{i_1,-1}$ and $s_{i_2,1}$.
\end{itemize}
We illustrate the nodes and involved strategies in \Cref{fig:PLSgadget}.
\begin{figure}
    \centering
    \includegraphics[width=0.7\linewidth]{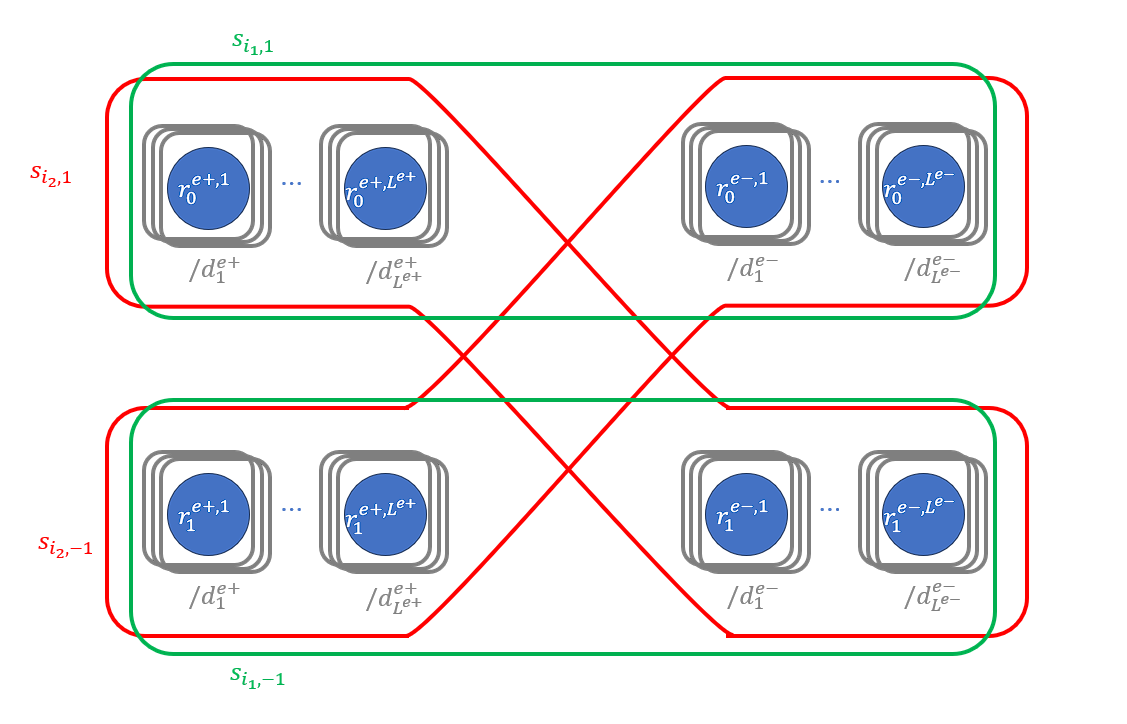}
    \caption{An illustration of the construction of $GAD_e$. The nodes are illustrated as blue circles. The grey rectangles represent the dummy agents attrcting each node, with each grey label $/d_k^{e+}$ and $/d_k^{e-}$ representing that the number of dummy agents on the node is $d_k^{e+}$ or $d_k^{e-}$ respectively. The two strategies of $i_1$ are represented as green boxes, and two strategies of $i_2$ are represented as red boxes.}
    \label{fig:PLSgadget}
\end{figure}

We properly design two series: $\{d^{e+}_k\}_{k=1,\cdots,L^{e+}}$ and $\{d^{e-}_k\}_{k=1,\cdots,L^{e-}}$ so that $GAD_e$ can represent the edge weight $W(e)$.

Define $\phi(j,\vec{S}(\vec{x}))=\sum_{i=1}^{c(j,\vec{S}(\vec{x}))}\frac{1}i$ for any $j\in N$, then $\Phi(\vec{S})=\sum_{j\in N}\phi(j,\vec{S})$. We know that $\vec{S}(\vec{x})$ is a PNE if and only if it is a local maximum of $\Phi$, i.e., $\Phi(\vec{S}(\vec{x}))$ cannot be increased by changing the value of any single $x_i$.

For each $k=1,\cdots,L^{e+}$, we calculate $\phi(r^{e+,k}_{0},\vec{S}(\vec{x}))+\phi(r^{e+,k}_{1},\vec{S}(\vec{x}))=\sum_{i=1}^{c(r^{e+,k}_{0},\vec{S}(\vec{x}))}\frac1{i}+\sum_{i=1}^{c(r^{e+,k}_{1},\vec{S}(\vec{x}))}\frac1{i}$. We have two cases:

(a) If $x_{i_1}x_{i_2}=1$, i.e., $x_{i_1}=x_{i_2}$, we know that either $c(r^{e+,k}_{0},\vec{S}(\vec{x}))=d^{e+}_k+2,c(r^{e+,k}_{1},\vec{S}(\vec{x}))=d^{e+}_k$, or $c(r^{e+,k}_{0},\vec{S}(\vec{x}))=d^{e+}_k,c(r^{e+,k}_{1},\vec{S}(\vec{x}))=d^{e+}_k+2$. Therefore $$\phi(r^{e+,k}_{0},\vec{S}(\vec{x}))+\phi(r^{e+,k}_{1},\vec{S}(\vec{x}))=\sum_{i=1}^{d^{e+}_k+2}\frac1{i}+\sum_{i=1}^{d^{e+}_k}\frac1{i}.$$

(b) If $x_{i_1}x_{i_2}=-1$, i.e. $x_{i_1}\neq x_{i_2}$, we know that  $c(r^{e+,k}_{0},\vec{S}(\vec{x}))=c(r^{e+,k}_{1},\vec{S}(\vec{x}))=d^{e+}_k+1$.
It follows that $$\phi(r^{e+,k}_{0},\vec{S}(\vec{x}))+\phi(r^{e+,k}_{1},\vec{S}(\vec{x}))=\sum_{i=1}^{d^{e+}_k+1}\frac1{i}+\sum_{i=1}^{d^{e+}_k+1}\frac1{i}.$$
Combining cases (a) and (b), we get
\begin{align*}
    &\phi(r^{e+,k}_{0},\vec{S}(\vec{x}))+\phi(r^{e+,k}_{1},\vec{S}(\vec{x}))\\
    =&\frac{1-x_{i_1}x_{i_2}}{2}(\frac1{d^{e+}_k+1}-\frac1{d^{e+}_k+2})+2\sum_{i=1}^{d^{e+}_k+1}\frac1{i}.
\end{align*}

For each $k=1,\cdots,L^{e-}$, by similar calculation, we have
\begin{align*}
    &\phi(r^{e-,k}_{0},\vec{S}(\vec{x}))+\phi(r^{e-,k}_{1},\vec{S}(\vec{x}))\\
    =&\frac{x_{i_1}x_{i_2}-1}{2}(\frac1{d^{e-}_k+1}-\frac1{d^{e-}_k+2})+2\sum_{i=1}^{d^{e-}_k+1}\frac1{i}.
\end{align*}

Let $N^{GAD_e}$ denote the set of all $2(L^{e+}+L^{e-})$ nodes in $GAD_e$. Let $\rho(e)=\sum_{k=1}^{L^{e+}}2\sum_{i=1}^{d^{e+}_k+1}\frac1{i}+\sum_{k=1}^{L^{e-}}2\sum_{i=1}^{d^{e-}_k+1}\frac1{i}$. By summing over $N^{GAD_e}$, we get
\begin{align*}\sum_{j\in N^{GAD_e}}\phi(j,\vec{S}(\vec{x}))
=&\sum_{k=1}^{L^{e+}}(\frac{1-x_{i_1}x_{i_2}}{2}(\frac1{d^{e+}_k+1}-\frac1{d^{e+}_k+2})+2\sum_{i=1}^{d^{e+}_k+1}\frac1{i})\\
+&\sum_{k=1}^{L^{e-}}(\frac{x_{i_1}x_{i_2}-1}{2}(\frac1{d^{e-}_k+1}-\frac1{d^{e-}_k+2})+2\sum_{i=1}^{d^{e-}_k+1}\frac1{i})\\
=&\frac{1-x_{i_1}x_{i_2}}{2}(\sum_{k=1}^{L^{e+}}(\frac1{d^{e+}_k+1}-\frac1{d^{e+}_k+2})-\sum_{k=1}^{L^{e-}}(\frac1{d^{e-}_k+1}-\frac1{d^{e-}_k+2}))+\rho(e)
\end{align*}

Our goal is to take a rational number $\lambda\in(0,1)$ and design $\{d^{e+}_k\}_{k=1,\cdots,L^{e+}}$ and $\{d^{e-}_k\}_{k=1,\cdots,L^{e-}}$ for each $e\in E$ so that \begin{equation}
    \lambda W(e)=\sum_{k=1}^{L^{e+}}(\frac1{d^{e+}_k+1}-\frac1{d^{e+}_k+2})-\sum_{k=1}^{L^{e-}}(\frac1{d^{e-}_k+1}-\frac1{d^{e-}_k+2}). \label{eq:PLS-gadget1}
\end{equation}
When \Cref{eq:PLS-gadget1} holds for all $e\in E$, we have 
\begin{align*}
\Phi(\vec{S}(\vec{x}))=&\sum_{e\in E}\sum_{j\in N^{GAD_e}}\phi(j,\vec{S}(\vec{x}))\\
=&\sum_{e=(i_1,i_2)\in E}\bigg(\frac{1-x_{i_1}x_{i_2}}{2}\bigg(\sum_{k=1}^{L^{e+}}\Big(\frac1{d^{e+}_k+1}-\frac1{d^{e+}_k+2}\Big) \\
&-\sum_{k=1}^{L^{e-}}\Big(\frac1{d^{e-}_k+1}-\frac1{d^{e-}_k+2}\Big)\bigg)+\rho(e)\bigg)\\
=&\sum_{e\in E}\frac{1-x_{i_1}x_{i_2}}{2}\lambda W(e)+\sum_{e\in E}\rho(e)\\
=&\lambda\cdot CUTWEIGHT(\vec{x})+\sum_{e\in E}\rho(e).
\end{align*}
This implies that $\Phi(\vec{S}(\vec{x}))$ is locally maximal if and only if $\vec{x}$ is a local maximum of $CUTWEIG$ $HT(\vec{x})$, i.e., $(V'(\vec{x}),V''(\vec{x}))$ is a local maximum cut. 


Now we only need to construct $\{d^{e+}_k\}_{k=1,\cdots,L^{e+}}$ and $\{d^{e-}_k\}_{k=1,\cdots,L^{e-}}$ for each $e\in E$. We demonstrate that this is possible in \Cref{lem:PLS-edge-gadget} based on the technical result \Cref{lem:PLS-decompose-fraction}, both of which will be proved later.
By \Cref{lem:PLS-edge-gadget}, let $\bar{W}=\max_{e\in E}W(e)$,  we can calculate $\lambda\in(0,1)$, so that for each $e\in E$, the two series $\{d^{e+}_k\}_{k=1,\cdots,L^{e+}}$ and $\{d^{e-}_k\}_{k=1,\cdots,L^{e-}}$ which satisfies \Cref{eq:PLS-gadget1} can be computed in polynomial time in $\log \bar{W}$. Note that the total number of nodes and dummy agents used to construct each $GAD_e$ are $2(L^{e+}+L^{e-})$ and $2(\sum_{k=1}^{L^{e+}}d^{e+,k}+\sum_{k=1}^{L^{e-}}d^{e-,k})$, respectively, and they are both bounded polynomially. Thus $|N|$, $|A|$, and $\sum_{i\in A}|\mathcal{S}_i|$ are polynomial in $|V|$, $|E|$, and $\log\bar{W}$.

To conclude, the instance $\mathcal{I}$ can be computed within polynomial time in the description length of $G$, and each PNE of $\mathcal{I}$ can be mapped to a local maximum cut of $G$. This completes the proof.
\end{proof}
\begin{lemma}\label{lem:PLS-edge-gadget}
Given an upper bound $\bar{W}\in\mathbb{N}^*$ on edge weights, there exists $\lambda\in(0,1)$, such that for any integer $w\in[0,\bar{W}]$, one can construct $\{d^{+}_k\}_{k=1,\cdots,L^{+}}$ and $\{d^{-}_k\}_{k=1,\cdots,L^{-}}$, such that $\lambda w=\sum_{k=1}^{L^{+}}(\frac1{d^{+}_k+1}-\frac1{d^{+}_k+2})-\sum_{k=1}^{L^{-}}(\frac1{d^{-}_k+1}-\frac1{d^{e-}_k+2})$. 
Moreover, $\sum_{k=1}^{L^{+}}(d^{+}_k+1)+\sum_{k=1}^{L^{-}}(d^{-}_k+1)$ is bounded by a polynomial in $\log \bar{W}$, and the construction can be computed in polynomial time.
\end{lemma}
\begin{proof}
First, observe that for any $b\in\mathbb{N}^*$, \begin{align*}
\frac1b=1-\sum_{i=1}^{b-1}\frac1{i(i+1)}=\frac1{1\cdot 2}+\frac1{1\cdot 2}-\sum_{i=0}^{b-2}\frac{1}{(i+1)(i+2)}.
\end{align*}
Thus any $\frac1{b}$ or $-\frac1{b}$ can be written as the sum of $b+2$ terms in the form of $\pm(\frac1{d+1}-\frac1{d+2})$.

According to \Cref{lem:PLS-decompose-fraction} which will be proved later, let $\bar{n}=\lceil\log_2(\bar{W})\rceil$, we can calculate $\lambda=\lambda_{\bar{n}}$. Given any $w\in [0,\bar{W}]$, we can calculate $\{(B_i,C_i)\}_{i=1,\cdots,L}$ such that $\sum_{i=1}^L\frac{C_i}{B_i}=w\lambda$. Each term $\frac{C_i}{B_i}$ can be represented by $|C_i|(B_i+2)$ terms in the form of $\pm(\frac1{d+1}-\frac1{d+2})$, where $d\leq \max\{B_i,2\}$. By collecting all these terms, we get $\{d^{+}_k\}_{k=1,\cdots,L^{+}}$ and $\{d^{-}_k\}_{k=1,\cdots,L^{-}}$ as desired.

Since $\sum_{i=1}^L|C_i|$ and $\max_{i=1,\cdots,L}B_i$ are bounded by a polynomial in $\bar{n}=O(\log\bar{W})$, we know that $\sum_{k=1}^{L^{+}}(d^{+}_k+1)+\sum_{k=1}^{L^{-}}(d^{-}_k+1)$ is polynomial in $\log\bar{W}$. It is easy to see that the construction is polynomial-time computable.
\end{proof}

\begin{lemma}\label{lem:PLS-decompose-fraction}
For any integer $n>0$, there exists $\lambda_n\in(0,1)$ such that, for any integer $w\in[0,2^n]$, there exist a series of numbers, $\{(B_i,C_i)\}_{i=1,\cdots,L}$, so that $$w\lambda_n=\sum_{i=1}^L \frac{C_i}{B_i},$$
where $B_i\in\mathbb{Z}_{\geq1},C_i\in\mathbb{Z}$, $L\in\mathbb{Z}_{\geq 1}$, and $\sum_{i=1}^L|C_i|$ and $\max_i B_i$ is bounded by a polynomial in $n$ \footnote{During the statement and proof of this lemma, polynomial means polynomial in $n$ when not specified.}. Moreover, $\lambda_n$ and $\{(B_i,C_i)\}_{i=1,\cdots,L}$ can be computed in polynomial time.
\end{lemma}
\begin{proof}
Let $p_i$ denote the $i$-th smallest prime number. By the prime number theorem, we have $p_n=O(n\log n)$. Therefore, given any $n>0$, $p_1,\cdots,p_n$ can be computed in polynomial time by simply enumerating all the natural numbers below $O(n\log n)$ and checking the primality of each number.
We let $M=\prod_{i=1}^np_i$ and take $\lambda_n=\frac1{M}$.

Define $r_1,\cdots,r_n$ such that $r_i=\prod_{i'\neq i}p_{i'}=\frac{M}{p_i}$. Note that $1$ is the only common factor of $r_1,\cdots,r_n$, so there exists $t_1,\cdots,t_n\in\mathbb{Z}$ such that $\sum_{i=1}^nt_ir_i\equiv 1\pmod{M}$. In other words, $\sum_{i=1}^nt_ir_i=1+TM$ for some $T\in\mathbb{Z}$.

For any integer $w\in[0,2^n]$, we have $$\sum_{i=1}^n\frac{wt_i}{p_i}=\sum_{i=1}^n\frac{wt_ir_i}{M}=\frac{w(1+TM)}M=\frac{w}M+wT.$$
Let $C_i'=(wt_i\bmod{p_i})$ for $i=1,\cdots,n$ and  $C''=\sum_{i=1}^n\frac{wt_i-C_i'}{p_i}-wT$, then one can see that $C''$ is an integer. In addition, $$C''+\sum_{i=1}^n\frac{C_i'}{p_i}=\sum_{i=1}^n\frac{wt_i}{p_i}-wT=\frac{w}{M}.$$
Note that $M\geq 2^n\geq w$, and $|C_i'|< p_i$ for any $i=1,\cdots,n$. Therefore $$|C''|=\left|\frac{w}{M}-\sum_{i=1}^n\frac{C_i'}{p_i}\right|\leq \left|\frac{w}{M}\right|+\sum_{i=1}^n\frac{|C_i'|}{p_i}\leq n+1.$$

Define $L=n+1$. For each $i=1,\cdots,n$, let $C_i=C'_i, B_i=p_i$, and let $C_{n+1}=C'', B_{n+1}=1$. Then, it follows that $\sum_{i=1}^L\frac{C_i}{B_i}=\frac{w}{M}=w\lambda_n$.

Since $|C_i|< p_i$, $B_i=p_i$ holds for $i=1,\cdots,n$, we have $\sum_{i=1}^L|C_i|\leq \sum_{i=1}^np_i+|C''|\leq np_n+n+1=O\left(n^2\log n\right)$ and $\max_{i=1,\cdots,L}B_i=O(n\log n)$,
which are polynomially bounded.

Moreover, we note that $t_1,\cdots,t_n$ can be constructed by iteratively applying the Extended Euclidean Algorithm. For $i=1,\cdots,n-1$, given $\gcd(r_1,\cdots,r_i)$, one can use the Extended Euclidean Algorithm to calculate $\gcd(r_1,\cdots,r_{i+1})=\gcd(\gcd(r_1,\cdots,r_i),r_{i+1})$, as well as two coefficients $x_i,y_{i+1}$ such that $x_i\cdot \gcd(r_1,\cdots,r_i)+y_{i+1}\cdot r_{i+1}=\gcd(r_1,\cdots,r_{i+1})$. Herein, the absolute values of $x_i$ and $y_{i+1}$ are bounded by $M$. Let $y_1=1$, we can take $t_i=(y_i\prod_{j=i}^{n-1}x_j\bmod{M})$, and it can be verified that $\sum_{i=1}^nt_ir_i\equiv 1\pmod{M}$. Since $\log M=O(n\log n)$, all the addition, multiplication and division with remainder operations can be done in polynomial time. The total number of operations is also polynomial. In a summary, $\lambda_n$ and $\{(B_i,C_i)\}_{i=1,\cdots,L}$ can be computed in time polynomial in $n$.
\end{proof}
So far, we have shown that finding a PNE of an ($\vec{\mathcal{S}}$)-asymmetric CAG is PLS-hard. Then, the following lemma assures that the equivalent hardness of finding a PNE between ($\vec{\mathcal{S}}$)-asymmetric CAG and symmetric CAG.
\begin{lemma}
\label{lem:PLS-unionizing-strategy}
Any instance of ($\vec{\mathcal{S}}$)-asymmetric CAG $\mathcal{I}=(N,A,\vec{\mathcal{S}},\vec{w},\vec{v})$ can be mapped to an instance of symmetric CAG $\mathcal{I}'=(N',A',\vec{\mathcal{S}}',\vec{w}',\vec{v}')$, such that any PNE of $\mathcal{I}'$ can be mapped back to a PNE of $\mathcal{I}$, where both mapping are polynomial-time computable.
\end{lemma}


\begin{proof}

Suppose $|N|=n,|A|=m$. Without loss of generality, assume the agents are labeled as $A=\{1,\cdots,m\}$. Note that $n$ is an upper bound of any agent's utility under any strategy profile in $\mathcal{I}$.

We construct $\mathcal{I}'$ as follows. Let $M=2n+1$. We create a group of $M$ nodes $r_{i,1},\cdots,r_{i,M}$ with big value for each $i\in A$. Let $N^i$ denote $\{r_{i,1},\cdots,r_{i,M}\}$ and $N'=N\cup\bigcup_{i=1}^m N^i$.

For each $i\in A$, we add $N^i$ to each of agent $i$'s strategy, and obtain $\bar{\mathcal{S}}_i=\{s\cup N^i:s\in\mathcal{S}_i\}$. Define the common strategy space $$\bar{\mathcal{S}}=\bigcup_{i\in A}\bar{\mathcal{S}}_i.$$
Let $\mathcal{S}'_i=\bar{\mathcal{S}}$ for all $i\in A'$, and $\vec{\mathcal{S}}'=\times_{i\in A'}\bar{\mathcal{S}}$.

It is obvious that the total computation time for constructing $\mathcal{I'}$ and splitting the nodes in $\mathcal{N'}$ is polynomial in $|N|$, $|A|$, and $\sum_{i\in A}|\mathcal{S}_i|$.

By the construction of $\bar{\mathcal{S}}$, we know that for any strategy profile $\vec{S}'\in\vec{\mathcal{S}}'$ and any agent $i'\in A'$, there is exactly one $i\in A$, such that $S'_{i'}\in\bar{\mathcal{S}}_i$. We denote such correspondence as $i=\mu(i',\vec{S}')$. It follows that $i'$ attracts the nodes $N^i$ under $\vec{S}'$ if and only if $\mu(i',\vec{S}')=i$. Intuitively, each agent $i'$ in $\mathcal{I}'$ takes on the role of agent $\mu(i',\vec{S}')$ in $\mathcal{I}$.

We say a strategy profile $\vec{S}'\in\vec{\mathcal{S}}'$ is perfectly-matched, if for all $i\in A$, there is exactly one $i'\in A'$, such that $\mu(i',\vec{S}')=i$.

Observe that any perfectly-matched strategy profile $\vec{S}'\in\vec{\mathcal{S}}'$ induces a strategy profile $\vec{S}\in\vec{\mathcal{S}}$ in $\mathcal{I}$, such that for any $i'\in A'$ and $i=\mu(i',\vec{S}')$, it holds that $$S'_{i'}=S_{i}\cup N^i.$$
For any $j\in N$, it can be seen that $c(j,\vec{S}')=c(j,\vec{S})$. Therefore, for any agent $i'\in A'$, we have
$$U_{i'}^{\mathcal{I}'}(\vec{S}')=M+\sum_{j\in N}\frac{1}{c(j,\vec{S}')}=M+U_{\mu(i',\vec{S}')}^{\mathcal{I}}(\vec{S}).$$

Given a PNE $\vec{S}^{*\prime}$ of $\mathcal{I}'$, we construct a PNE $\vec{S}^*$ of $\mathcal{I}$. 
We first show that $\vec{S}^{*\prime}$ must be perfectly-matched. Suppose for the sake of contradiction, that $\vec{S}^{*\prime}$ is not perfectly-matched. Then there exists $i_1,i_2\in A$, $i'_1,i'_2\in A'$, such that $i'_1\neq i'_2$, $\mu(i'_1,\vec{S}^{*\prime})=\mu(i'_2,\vec{S}^{*\prime})=i_1$, and $\mu(i',\vec{S}^{*\prime})\neq i_2$ for all $i'\in A'$. It implies that  $c(r_{i_2,k},\vec{S}^{*\prime})=0$ for all $k=1,\cdots,M$. Observe that $U_{i'_1}^{\mathcal{I}'}(\vec{S}^{*\prime})\leq \frac1{2}M+n<M$. If $i'_1$ deviates to an arbitrary strategy in $\bar{\mathcal{S}}_{i_2}$, its utility will improve to at least $M$, which contradicts with the assumption that $\vec{S}^{*\prime}$ is a PNE. Therefore $\vec{S}^{*\prime}$ is perfectly-matched.

Let $\vec{S}^*$ be the strategy profile in $\mathcal{I}$ induced by $\vec{S}^{*\prime}$, we prove that $\vec{S}^*$ is a PNE by contradiction. Suppose that there is $i\in A$ and $s_i\in\mathcal{S}_i$ such that $U_{i}^{\mathcal{I}}(\vec{S}_{s_i})> U_{i}^{\mathcal{I}}(\vec{S}^*)$, where $\vec{S}_s=(s_i,\vec{S}^*_{-i})$. 
Take $i'\in A'$ such that $\mu(i',\vec{S}^{*\prime})=i$ and let $\vec{S}'_{s}=(s\cup N^i,\vec{S}^{*\prime}_{-i'})$ denote the strategy profile obtained by changing the strategy of $i'$ to $s\cup N^i$ in $\vec{S}^{*\prime}$. Since both $\vec{S}^{*\prime}$ and $\vec{S}'_{s}$ are perfectly-matched, we obtain
$$U_{i'}^{\mathcal{I}'}(\vec{S}'_{s})=M+U_{i}^{\mathcal{I}}(\vec{S}_{s})>M+U_{i}^{\mathcal{I}}(\vec{S}^*)=U_{i'}^{\mathcal{I}'}(\vec{S}^{*\prime}),$$
which contradicts with the assumption that $\vec{S}^{*\prime}$ is a PNE of $\mathcal{I}$. 
Note that the construction of $\vec{S}^*$ from any $\vec{S}^{*\prime}$ is also polynomial-time computable.
\end{proof}

In this way, we prove that for a symmetric CAG, finding a pure Nash equilibrium is PLS-hard. This completes the proof of \Cref{thm:PLS-complete}.
\end{proof}

Since the problem of finding a PNE is PLS-complete, we turn to designing the efficient algorithm to compute the approximate PNE. Fortunately, for any $\epsilon>0$, a $(1+\epsilon)$-approximate pure Nash equilibrium can be found in polynomial time by the best-response dynamics. 
That is, in each step, if the current strategy profile $\vec{S}$ is not a $(1+\epsilon)$-approximate PNE, an agent $i$ is chosen and deviates her strategy to $S'_i$, such that $$U_i(S'_i,\vec{S}_{-i})-U_i(\vec{S})=\max_{{i'}\in A,S'_{i'}\in \mathcal{S}_{i'}}(U_{i'}(S'_{i'},\vec{S}_{-i'})-U_{i'}(\vec{S})).$$ We provide the detailed proof in appendix. 

\begin{theorem}\label{thm:fptas}
Given $0<\epsilon<1$, for any instance $\mathcal{I}=(N,A,\vec{\mathcal{S}},\vec{w},\vec{v})$ of symmetric CAG, the best-response dynamics finds an $\epsilon$-pure Nash equilibrium in $O(\epsilon^{-1}|N||A|\log{|A|})$ steps, which implies that it is an FPTAS to compute an $(1+\epsilon)$-approximate PNE.
\end{theorem}
\begin{proof}
Let $M=|N|\sum_{k=1}^{|A|}\frac1k$. Thus, for any $\vec{S}\in\vec{\mathcal{S}}$, the potential function is upper bounded by $M$: $\Phi(\vec{S})\leq\sum_{j\in N}\sum_{k=1}^{|A|}\frac1k=M.$

Next, without any loss of generality, we assume that each agent has at least one strategy, and each strategy is non-empty. Then the utility of any agent is lower bounded: for any $\vec{S}\in\vec{\mathcal{S}}$ and any $i\in A$, $U_i(\vec{S})=\sum_{j\in S_i}\frac{1}{c(j,\vec{S})}\geq \frac1{|A|}$.
 

Assume that the current $\vec{S}$ is not an $\epsilon$-PNE, we know that there exists $i'\in A$ and $S'_{i'}\in \mathcal{S}_{i'}$ such that $U_{i'}(S'_{i'},\vec{S}_{-i'})>(1+\epsilon) U_{i'}(\vec{S})\geq U_{i'}(\vec{S})+\frac{\epsilon}{|A|}$. By the rule of best-response dynamics, the benefit of the chosen agent satisfies $U_i(\vec{S}')-U_i(\vec{S})=\max_{{i'}\in A,S'_{i'}\in \mathcal{S}_{i'}}(U_{i'}(S'_{i'},\vec{S}_{-i'})$ $-U_{i'}(\vec{S}))\geq U_{i'}(S'_{i'},\vec{S}_{-i'})-U_{i'}(\vec{S})>\frac{\epsilon}{|A|}$.

By the definition of potential game, we have $\Phi(\vec{S}')-\Phi(\vec{S})=U_i(\vec{S}')-U_i(\vec{S})>\frac{\epsilon}{|A|}$. Since $0\leq \Phi(\vec{S})\leq M$ for any $\vec{S}\in\mathcal{S}$, this happens no more than $M/\frac{\epsilon}{|A|}=O(\epsilon^{-1}|N||A|\log{|A|})$ times.
In conclusion, the best-response dynamics from any initial strategy profile finds an $\epsilon$-pure Nash equilibrium in $O(\epsilon^{-1}|N||A|\log{|A|})$ steps.
\end{proof}

In the rest of this section, we concentrate on the efficiency of PNE, i.e., the PoA in symmetric setting. Based on the results of literature \cite{V02}, it is not difficult to check that customer attraction games (even in asymmetric case) belong to the valid utility system introduced by \cite{V02}, which implies that the upper bound of PoA is 2 directly. On the other hand, we also use another technique to prove a more detailed upper bound, $\min\{2, $ $\max\{n/m,1\}\}$, on any instance of symmetric CAG. Combining with the examples of lower bound, we finally show that the tight PoA is 2 for the customer attraction games.
\begin{theorem}\label{lem:up_poa}
For any instance $\mathcal{I}=(N,A,\vec{\mathcal{S}},\vec{w},\vec{v})$ of symmetric CAG, given $|N|=n$ and $|A|=m$, the price of anarchy is $\min\{2, $ $\max\{n/m,1\}\}$. Generally, the price of anarchy for CAG is tight and equal to 2.
\end{theorem}
\begin{proof}
    We first show the upper bound of 2. Denote the optimal profile and an NE profile by $\vec{S}^*$ and $\vec{S}^{NE}$, respectively. Then, we have
        \begin{align}
            \text{SW}(\vec{S}^*) & = \sum_{i=1}^{m} U_i(\vec{S}^*) 
                 \leq \sum_{i=1}^{2m} U_i\big((\vec{S}^*,\vec{S}^{NE})\big) \label{ieq:add NE to opt}\\
                & \leq \sum_{i=1}^m U_i(\vec{S}_i^*,\vec{S}^{NE}_{-i}) + \sum_{i=1}^m U_i(\vec{S}^{NE}) \label{ieq: split 2K}\\
                & \leq \sum_{i=1}^m U_i(\vec{S}^{NE}) + \sum_{i=1}^m U_i(\vec{S}^{NE}) = 2\text{SW}(\vec{S}^{NE}). \nonumber
        \end{align}
    The inequality (\ref{ieq:add NE to opt}) holds because we add $m$ more agents and let them choose the NE profile, which will attract more nodes. The inequality (\ref{ieq: split 2K}) is based on that we split $2m$ agents into two groups and calculate their utility in each group with only $m$ agents, which does not make utility decrease. The last inequality is due to the property of NE. 
    
    
    Next we show the upper bound of $\max\{n/m,1\}$. 
    When $n/m \leq 1$, since the number of agents is greater than the number of nodes, it is easy to check that all nodes are attracted in the optimal and NE profile, which means that $\text{PoA}(\mathcal{I})= 1$.  
    When $n/m > 1$, w.l.o.g, suppose $\text{PoA}(\mathcal{I}) > 1$. There must exist a PNE $\vec{S}^{NE}$, such that at least one node $j\in N$ is not attracted, i.e., $c(j,\vec{S}^{NE})=0$. Take any $S'_i\in\mathcal{S}$ such that $j\in S'_i$. Since $\vec{S}^{NE}$ is a Nash equilibrium, for any agent $i\in A$, we have $U_i(\vec{S}^{NE})\geq U_i(S'_i,\vec{S}^{NE}_{-i})\geq 1$. Therefore, we have $\text{SW}(\vec{S}^{NE})=\sum_{i\in A}U_i(\vec{S}^{NE})\geq m$.
    On the other hand, we have $\text{SW}(\vec{S}^*)\leq |N|=n$. Thus, we know that either $\text{PoA}(\mathcal{I})=1$ or $\text{PoA}(\mathcal{I})\leq n/m$, so $\text{PoA}(\mathcal{I})\leq\max\{1,n/m\}$.

We have shown that when $n/m<1$, the PoA is equal to $1$. To prove the tightness of the above bound of PoA, we still provide two examples for constructing the lower bounds of $2$ and $n/m$.
\begin{lemma}\label{lem:lb_poa}
    For the symmetric CAG, we can find two instances whose PoA are exactly $2$ and $n/m$, respectively. 
\end{lemma}
\begin{proof}
For any $n$ and $m$ such that $n>m$, we construct the following instance with $n$ nodes and $m$ agents. Define the node set and agent set as $N=\{q_1,q_2,\cdots,q_n\}$ and $A=\{1,\cdots,m\}$, respectively. Let $\mathcal{S}=\{\{q_1,\cdots,q_m\},\{q_{m+1}\},\{q_{m+2}\},\cdots,$ $\{q_{n}\}\}$ and $\mathcal{S}_i=\mathcal{S}$ for all $i\in A$. 

When $n\leq 2m-1$, since $|\mathcal{S}|=n-m+1\leq m$, all sets are selected in the optimal profile,,and we have $\text{SW}(\vec{S^*})=n$. When $n\geq 2m$, the social welfare is maximized by letting one agent select $\{q_1,\cdots,q_m\}$ and each of the other agents select a different node, respectively. We have $\text{SW}(\vec{S^*})=m+m-1=2m-1$.

Let $\vec{S}^{NE}$ be the strategy profile where $S^{NE}_i=\{q_1,\cdots,q_m\}$ for all $i\in A$. It is obvious that $\vec{S}^{NE}$ is a PNE and $\text{SW}(\vec{S}^{NE})=m$. Therefore, the PoA is $n/m$ when $n<2m$ and $(2m-1)/m$ when $n\geq 2m$.
\end{proof}
In summary, the bound of PoA, $\min\{2,\max\{n/m,1\}\}$, is tight. In general, for the symmetric CAG, the PoA is $2$ over all instances.  
\end{proof}
In reality, some instances assuring $n/m<2$ can achieve a PoA better than 2. When $n/m\geq 2$, the worst PNE of above example leads to one half loss of social welfare, which reflects the competition phenomena where agents pursue the hot topics in the red and blue ocean markets, introduced in \Cref{Sec:Intro}. 

\section{The Asymmetric Static Game}
\label{Sec:asymmetric}

In this section, we begin to investigate the asymmetric CAG. Compared to some elegant results in symmetric CAG (e.g., the existence of PNE is guaranteed under the symmetric CAG, as well as any $(\vec{\mathcal{S}},\vec{v})$-asymmetric CAG where the weights of agents are restricted to be symmetric, yet), the general asymmetric case becomes trickier, which is due to that the potential function is no longer applicable. In fact, a $(\vec{w})$-asymmetric CAG is generally not an exact potential game, as shown in the following lemma. 
\begin{lemma}\label{lem:not potential}
There is an instance $\mathcal{I}=(N,A,\vec{\mathcal{S}},\vec{w},\vec{v})$ of $(\vec{w})$-asymmetric CAG, such that $\mathcal{I}$ is not an exact potential game.
\end{lemma} 

\begin{proof}
Consider the following asymmetric instance $\mathcal{I}$: there is one node $N=\{q_1\}$ with the value $v_{q_1}=1$, and two agents $A=\{1,2\}$ with the weights $\vec{w}=(1,2)$ and the same strategy space $\mathcal{S}_1=\mathcal{S}_2=\{N,\emptyset\}$. 

We show this lemma by contradiction. Assume that $\mathcal{I}$ is an exact potential game with respect to a potential function $\Phi(S_1,S_2)$. By the definition of exact potential game, we have 
\begin{align*}
    &U_1(N,S_2)-U_1(\emptyset,S_2)=\Phi(N,S_2)-\Phi(\emptyset,S_2), ~~ \forall S_2\in\{N,\emptyset\}, \\
    \text{and}~~ &U_2(S_1,N)-U_2(S_1,\emptyset)=\Phi(S_1,N)-\Phi(S_1,\emptyset), ~~\forall S_1\in\{N,\emptyset\}. 
\end{align*}
On the one hand, we know that $U_1(N,\emptyset)-U_1(\emptyset,\emptyset)=1$ and $U_2(N,N)-U_2(N,\emptyset)=\frac2{3}$. Therefore, we obtain 
\begin{align*}
    &\Phi(N,N)-\Phi(\emptyset,\emptyset)=[\Phi(N,N)-\Phi(N,\emptyset)]+[\Phi(N,\emptyset)-\Phi(\emptyset,\emptyset)]\\
    = &[U_2(N,N)-U_2(N,\emptyset)]+[U_1(N,\emptyset)-U_1(\emptyset,\emptyset)]=\frac53.
\end{align*}
On the other hand, we have that $U_2(\emptyset,N)-U_2(\emptyset,\emptyset)=1$ and $U_1(N,N)-U_1(\emptyset,N)=\frac1{3}$. Similarly, we get 
\begin{align*}
&\Phi(N,N)-\Phi(\emptyset,\emptyset)=[\Phi(N,N)-\Phi(\emptyset,N)]+[\Phi(\emptyset,N)-\Phi(\emptyset,\emptyset)]\\
=&[U_1(N,N)-U_1(\emptyset,N)]+[U_2(\emptyset,N)-U_2(\emptyset,\emptyset)]=\frac43,
\end{align*}
which is a contradiction. Consequently, $\mathcal{I}$ is not an exact potential game.
\end{proof}

A natural question is, whether the PNE is still guaranteed to exist when the agents have asymmetric weights. Interestingly, we find that when there are only two agents, PNE still exists by Theorem \ref{thm:exist asym two agent}. 
\begin{theorem}\label{thm:exist asym two agent}
    For any instance $\mathcal{I}$ of asymmetric customer attraction game with two agents, a pure Nash equilibrium always exists.
\end{theorem}
\begin{proof}
    To show the existence of PNE, we construct a new potential function to help us. Define the weights of agent 1 and 2 as $w_1$ and $w_2$, respectively. Given a strategy profile $\Vec{S}=(S_1,S_2)$, for any node $j$, define $h_j(\Vec{S})$ as
    \begin{equation*}
    h_j(\Vec{S})=\begin{cases}
        v_j(w_1+w_2-\frac{w_1w_2}{w_1+w_2}), & \text{if} ~ j\in S_1 ~ \text{and} ~ j\in S_2 \\
        v_jw_1, & \text{if} ~ j\in S_1 ~ \text{and} ~ j\notin S_2 \\
        v_jw_2, & \text{if} ~ j\notin S_1 ~ \text{and} ~ j\in S_2 \\
        0, & \text{Otherwise.}
        \end{cases}
\end{equation*}
Define the potential function $H(\Vec{S})=\sum_j h_j(\Vec{S})$. W.l.o.g., assume that agent 1 changes strategy to $S'_1$ and denote $\Vec{S}'= (S'_1, S_2)$ by the current strategy profile. We calculate the difference from $h_j(\Vec{S}')$ and $h_j(\Vec{S})$:
\begin{equation*}
    h_j(\Vec{S}')-h_j(\Vec{S})=\begin{cases}
        v_jw_1, & \text{if} ~ j\in S'_1 \backslash (S_1 \cup S_2)  \\
        -v_jw_1, & \text{if} ~ j\in S_1 \backslash (S'_1 \cup S_2) \\
        v_j(w_1-\frac{w_1w_2}{w_1+w_2}), & \text{if} ~ j\in (S'_1\cap S_2) \backslash S_1 \\
        -v_j(w_1-\frac{w_1w_2}{w_1+w_2}), & \text{if} ~ j\in (S_1\cap S_2) \backslash S'_1 \\
        0, & \text{otherwise.}
        \end{cases}
\end{equation*}
It is not hard to check that $H(\Vec{S}')-H(\Vec{S})= w_1(U_1(\Vec{S}')-U_1(\Vec{S}))$ and the similar result also holds for agent 2. Since $w_1,w_2\geq 0$, it implies that when the agent deviates for better utility, the potential function will be improved. Due to the finite space of strategy profiles, PNE always exists.
\end{proof}

However, once that the number of agents increases to three, we can construct a counterexample to demonstrate that the nonexistence of PNE can be caused by a single weighted agent, i.e., the PNE may not exist, even when all but one of the agents are identically weighted $1$, and the strategy spaces and node values are symmetric. 
To build such example, we first construct an instance of $(\vec{\mathcal{S}},\vec{w},\vec{v})$-asymmetric CAG (the fully asymmetric model) in \Cref{exa:twoagent-noPNE}, which has no pure Nash equilibrium. Then we convert it to an instance of $(\vec{w})$-asymmetric CAG, preserving the nonexistence of PNE.

Intuitively, to construct a fully asymmetric model without PNE, we start with two active agents with different weights to form a two-agent $2\times 2$ matrix game. Then, by adding some dummy agents, we can adjust the payoff matrix of these two agents so that the beneficial deviations form a cycle (This action is feasible because the current game is no longer a potential game), which avoids the existence of a PNE. 
\begin{example}[Nonexistence of PNE]
\label{exa:twoagent-noPNE}

Under the $(\vec{\mathcal{S}},\vec{w},\vec{v})$-asymmetric CAG, we construct an instance $\mathcal{I}=(N,A,$ $\vec{\mathcal{S}},\vec{w},\vec{v})$. There are $4$ nodes labeled as $N=\{q_1,q_2,q_3,q_4\}$. Let  $v_{q_1}=2,v_{q_2}=1,$ $v_{q_3}=1$ and $v_{q_4}=2$. There are $3$ agents, labeled as $A=\{1,2,3\}$, including two active agents 1,2 and one dummy agent 3. The weights of agents are $\vec{w}=(4,1,1)$. Both active agents $1$ and $2$ have two strategies: $s_{1,1}=\{q_1,q_2\},s_{1,2}=\{q_3,q_4\},s_{2,1}=\{q_1,q_3\},s_{2,2}=\{q_2,q_4\}$, i.e., the strategy spaces of the active agents are $\mathcal{S}_1=\{s_{1,1},s_{1,2}\}$ and $\mathcal{S}_2=\{s_{2,1},s_{2,2}\}$. Dummy agent 3 has one strategy $\mathcal{S}_3=\{\{q_{1},q_4\}\}$. Since agent 3 is dummy, the game can be viewed as a two-player game between agents $1$ and $2$, each of whom has two strategies. We calculate the utilities of agents $1$ and $2$, presented as a payoff matrix in \Cref{tab:nonexistence-PNE-payoff}. The weights and values are properly chosen so that agent $1$ has incentive to deviate in states $(s_{1,1},s_{2,1})$ and $(s_{1,2},s_{2,2})$, while agent $2$ has incentive to deviate in states $(s_{1,1},s_{2,2})$ and $(s_{1,2},s_{2,1})$. Therefore, there is no PNE in this instance.
\begin{table}[htbp]
\centering
\begin{tabular}{|c|c|c|}
\hline
\diagbox{$S_1$}{$U_1,U_2$}{$S_2$}& $s_{2,1}$ & $s_{2,2}$\\\hline
$s_{1,1}$ &
(7/3, 4/3) & (2.4, 1.2) \\ \hline
$s_{1,2}$ &
(2.4, 1.2) & (7/3, 4/3) \\ \hline
\end{tabular}
\caption{The payoff matrix of the two-player game between the two active agents}
\label{tab:nonexistence-PNE-payoff}
\end{table}
\end{example}
Then, we convert the instance in Example \ref{exa:twoagent-noPNE} to an instance of $(\vec{w})$-asymmetric CAG, while preserving the nonexistence of PNE. Intuitively, to obtain symmetric strategy spaces, 
we can modify the original strategy space for each agent by creating a large group of new nodes and adding them to every strategy of this agent, and then join all resulting strategy spaces together to get a common strategy space. When the total value of each large group is designed properly, each agent in the constructed instance will play the role of one agent in the original instance, and thus the constructed instance will be equivalent to the original instance with respect to the existence of PNE. To obtain symmetric node values, for any node $j$ whose value is greater than one, we can simply split node $j$ into a group of nodes with unit value. The summarized results are in the following lemma. 


\begin{lemma}
\label{lem:convert-to-symmetric-sandv}
Given any instance $\mathcal{I}=(N,A,\vec{\mathcal{S}},\vec{w},\vec{v})$ of $(\vec{\mathcal{S}},\vec{w},\vec{v})$-asymmetric CAG, if there is one agent $i\in A$ such that $w_i\geq 1$ and $w_{i'}=1$ for all $i'\neq i$, then an instance $\mathcal{I}'=(N',A',\vec{\mathcal{S}}',\vec{w}',\vec{v}')$ of $(\vec{w})$-asymmetric CAG can be computed, such that $\mathcal{I}$ has a PNE if and only if $\mathcal{I}'$ has a PNE. The computation time is polynomial in $|N|$, $|A|$, $\sum_{i\in A}|\mathcal{S}_i|$, $\max_{i\in A}w_i$, and $\sum_{j\in N}v_j$. Moreover, the construction can preserve the weights of agents, in other words, $A'=A$ and $\vec{w}'=\vec{w}$.
\end{lemma}
\begin{proof}
Without loss of generality, we assume that the agents are relabeled as $A=\{1,\cdots,|A|\}$, such that $w_1\geq 1$ and $w_i=1$ for $i=2,\cdots,|A|$. 

Note that for any instance with different values of nodes, we can always split each node into a group of nodes with unit value: for any node $j\in N$, we can replace it by $v_j$ nodes with value $1$, and replace the appearance of node $j$ in any agent's strategy by this group of new nodes. It can be seen easily that this action does not change the utility of any agent.  Therefore, we only need to construct an instance $\mathcal{I}'=(N',A',\vec{\mathcal{S}}',\vec{w}',\vec{v}')$ under $(\vec{w},\vec{v})$-asymmetric CAG, and then split all the nodes in $N'$ to get an instance of $(\vec{w})$-asymmetric CAG.

Let $M=\max_{i\in A}\max_{s\in\mathcal{S}_i}~\sum_{j\in s}v_j$ as an upper bound on the utility gained by any strategy in $\mathcal{I}$, then we have $M\leq \sum_{j\in N}v_j$. Let $T=w_1\in\mathbb{Z}_{\geq 1}$, and $M'=(T+1)(M+1)$. We construct $\mathcal{I}'=(N',A',\vec{\mathcal{S}}',\vec{w}',\vec{v}')$ as follows:
\begin{itemize}
    \item Keep the agent set unchanged, i.e., $A'=A$.
    \item For each agent $i\in A$, we create a new node $r_i$. Define the new node set as $N'=N\cup\{r_i:i\in A\}$.
    \item Define a new common strategy space as $$\bar{\mathcal{S}}=\bigcup_{i\in A}\bar{\mathcal{S}}_i,$$
    where $\bar{\mathcal{S}}_i=\{s\cup\{r_i\}:s\in\mathcal{S}_i\}$ for any $i\in A$. Define the common strategy space $\vec{\mathcal{S}}'=\times_{i\in A'}\bar{\mathcal{S}}$.
    \item Keep the weights of agents unchanged, i.e., $\vec{w}'=\vec{w}$.
    \item For each node $j\in N$, keep its value unchanged, that is, $v'_j=v_j$; For node $r_1$, set its value as $v'_{r_1}=(2T+1)M'$; For any node $r_i$, $i=2,\cdots,|A|$, set its value as $v'_{r_i}=2M'$.
\end{itemize}
It is obvious that the total computation time for constructing $\mathcal{I'}$ and splitting the nodes in $\mathcal{N'}$ is polynomial in $|N|$, $|A|$, $\sum_{i\in A}|\mathcal{S}_i|$, $\max_{i\in A}w_i$, and $\sum_{j\in N}v_j$.

By the construction of $\bar{\mathcal{S}}$, we know that for any strategy profile $\vec{S}'\in\vec{\mathcal{S}}'$ and any agent $i'\in A'$, there is exactly one $i\in A$, such that $S'_{i'}\in\bar{\mathcal{S}}_i$. We denote such correspondence as $i=\mu(i',\vec{S}')$. It follows that $i'$ attracts the node $r_{i}$ under $\vec{S}'$ if and only if $i=\mu(i',\vec{S}')$. Intuitively, each $i'$ in $\mathcal{I}'$ plays the role of agent $\mu(i',\vec{S}')$ in $\mathcal{I}$.

We say a strategy profile $\vec{S}'\in\vec{\mathcal{S}}'$ is perfectly-matched, if for all $i\in A$, there is exactly one $i'\in A'$ such that $\mu(i',\vec{S}')=i$, and $w_{i'}=w_i$. Observe that any perfectly-matched strategy profile $\vec{S}'\in\vec{\mathcal{S}}'$ induces a strategy profile $\vec{S}\in\vec{\mathcal{S}}$ in $\mathcal{I}$, such that for any $i'\in A'$ and $i=\mu(i',\vec{S}')$, it holds that $$S'_{i'}=S_{i}\cup\{r_i\}.$$
For any $j\in N$, it can be seen that $c(j,\vec{S}')=c(j,\vec{S})$. Therefore, for any agent $i'\in A'$, we have
$$U_{i'}^{\mathcal{I}'}(\vec{S}')=v'_{r_{\mu(i',\vec{S}')}}+\sum_{j\in N}v_{j}\frac{w_{i'}}{c(j,\vec{S}')}=v'_{r_{\mu(i',\vec{S}')}}+U_{\mu(i',\vec{S}')}^{\mathcal{I}}(\vec{S}).$$

Now we show that $\mathcal{I}'$ has a PNE if and only if $\mathcal{I}$ has a PNE.

\noindent\textbf{Sufficiency:} 

Suppose $\vec{S}^*\in\vec{\mathcal{S}}$ is a PNE of $\mathcal{I}$. Construct $\vec{S}^{*\prime}\in \vec{\mathcal{S}}'$ such that for each $i\in A'$, ${S}^{*\prime}_i=S^*_i\cup\{r_i\}$. We prove that $\vec{S}^{*\prime}$ is also a PNE of $\mathcal{I}'$.

Note that since $\mu(i',\vec{S}^{*\prime})=i'$ for all $i'\in A'$, we know that $\vec{S}^{*\prime}$ is perfectly-matched and $\vec{S}^*$ is the strategy profile in $\mathcal{I}$ induced by $\vec{S}^{*\prime}$. Therefore, we have that for any $i'\in A'$,
$$U_{i'}^{\mathcal{I}'}(\vec{S}^{*\prime})=v'_{r_{i'}}+U_{i'}^{\mathcal{I}}(\vec{S}^*).$$

For any agent $i'\in A'$ and any $s\in\bar{\mathcal{S}}$, we show that agent $i'$ cannot improve its utility by deviating to $s$. Let $\vec{S}'_{s}=(s,\vec{S}^{*\prime}_{-i'})$ denote the strategy profile replacing $i'$'s strategy by $s$ in $\vec{S}^{*\prime}$. If $r_{i'}\in s$, $\vec{S}'_{s}$ is still perfectly-matched and the strategy profile in $\mathcal{I}$ induced by $\vec{S}'_{s}$ is $\vec{S}_s=(s\setminus\{r_{i'}\},\vec{S}^*_{-i})$. Since $\vec{S}^*$ is a PNE, we have $U_{i'}^{\mathcal{I}}(\vec{S}^*)\geq U_{i'}^{\mathcal{I}}(\vec{S}_s)$ and therefore $U_{i'}^{\mathcal{I}'}(\vec{S}^{*\prime})\geq U_{i'}^{\mathcal{I}'}(\vec{S}'_{s})$. 

If $r_{i'}\notin s$, take $i\in A$ such that $r_i\in s$. Then, there are three cases:

(a) If $i'=1$ and $i\neq 1$, we have
$$U_{i'}^{\mathcal{I}'}(\vec{S}'_s)\leq 2M'+M<(2T+1)M'\leq U_{i'}^{\mathcal{I}'}(\vec{S}^{*\prime});$$

(b) If $i'\neq 1$ and $i=1$, we have $$U_{i'}^{\mathcal{I}'}(\vec{S}'_s)\leq \frac{1}{T+1}(2T+1)M'+M< (\frac{2T+1}{T+1}+\frac{1}{T+1})M'=2M'\leq U_{i'}^{\mathcal{I}'}(\vec{S}^{*\prime});$$

(c) If $i'\neq 1$ and $i\neq 1$, we have $$U_{i'}^{\mathcal{I}'}(\vec{S}'_s)\leq \frac{1}{1+1}2M'+M< 2M'\leq U_{i'}^{\mathcal{I}'}(\vec{S}^{*\prime}).$$
In summary, $U_{i'}^{\mathcal{I}'}(\vec{S}'_s)\leq U_{i'}^{\mathcal{I}'}(\vec{S}^{*\prime})$ holds in all cases. Therefore, by definition, we obtain that $\vec{S}^{*\prime}$ is a PNE of $\mathcal{I}'$.

\noindent\textbf{Necessity:} 

Suppose $\vec{S}^{*\prime}$ is a PNE of $\mathcal{I}'$, we construct a PNE $\vec{S}^*$ of $\mathcal{I}$. Firstly, we show that $\vec{S}^{*\prime}$ must be perfectly-matched. Suppose for the sake of contradiction, that $\vec{S}^{*\prime}$ is not perfectly-matched. Then at least one of the three cases below holds:

(a) There exists $i'_1,i'_2\in A'$, $i'_1\neq i'_2$, such that $\mu(i'_1,\vec{S}^{*\prime})=\mu(i'_2,\vec{S}^{*\prime})=i_1$ for some $i_1\in A$, and $w'_{i'_1}=w'_{i'_2}=w_{i_1}=1$. In this case, there must exist some $i_2\in A$ such that $c(r_{i_2},\vec{S}^{*\prime})=0$. Since the current utility of $i'_1$ is $$U_{i'_1}^{\mathcal{I}'}(\vec{S}^{*\prime})\leq 2M'\frac{1}{1+1}+M<2M',$$
$i'_1$ can deviate to any strategy attracting $r_{i_2}$ to improve its utility to at least $2M'$. This contradicts with the assumption that $\vec{S}^{*\prime}$ is a PNE.

(b) $T\geq 2$, and $\mu(1,\vec{S}^{*\prime})\neq 1$. In this case, if agent $1$ deviates to an arbitrary strategy in $\bar{\mathcal{S}}_1$, she will get a utility of at least $$v'_{r_1}\frac{T}{c(r_1,\vec{S}^{*\prime})+T}=\frac{(2T+1)T}{c(r_1,\vec{S}^{*\prime})+T}M'.$$
Since $\vec{S}^{*\prime}$ is a PNE, we have $$\frac{(2T+1)T}{c(r_1,\vec{S}^{*\prime})+T}M'\leq U_{1}^{\mathcal{I}'}(\vec{S}^{*\prime})\leq 2M'+M.$$
This implies that $$c(r_1,\vec{S}^{*\prime})\geq \frac{(2T+1)TM'}{2M'+M}-T=(\frac{(2T+1)(T+1)}{2(T+1)+1}-1)T.$$ Since $T\geq 2$, with a little calculation, we have 
$c(r_1,\vec{S}^{*\prime})>T$. Therefore, it holds that  $|\{i'\in A':\mu(i',\vec{S}^{*\prime})=1\}|=c(r_1,\vec{S}^{*\prime})\geq T+1$. There must exist some $i_1\in A$ such that $c(r_{i_1},\vec{S}^{*\prime})=0$. Take an arbitrary $i'_1\in A'$ such that $\mu(i'_1,\vec{S}^{*\prime})=1$, we get $$U_{i'_1}^{\mathcal{I}'}(\vec{S}^{*\prime})\leq (2T+1)M'\frac1{T+1}+M<2M'.$$ 
Therefore, $i'_1$ can deviate to any strategy attracting $r_{i_1}$ to improve the utility to at least $2M'$, which is a contradiction.

(c) $T\geq 2$, and there exists $i'_1\in A'\setminus\{1\}$ such that $\mu(i'_1,\vec{S}^{*\prime})=1$. In this case, there must exist some $i_1\in A$ such that $c(r_{i_1},\vec{S}^{*\prime})=0$. Since the current utility of $i'_1$ is 
$$U_{i'_1}^{\mathcal{I}'}(\vec{S}^{*\prime})\leq (2T+1)M'\frac1{T+1}+M<2M',$$ 
$i'_1$ can deviate to any strategy attracting $r_{i_1}$ to improve its utility to at least $2M'$, which is a contradiction.

In conclusion, we know that $\vec{S}^{*\prime}$ is perfectly-matched. Let $\vec{S}^*$ be the strategy profile in $\mathcal{I}$ induced by $\vec{S}^{*\prime}$, we prove that $\vec{S}^*$ is a PNE by contradiction. Suppose that there is $i\in A$ and $s\in\mathcal{S}_i$ such that $U_{i}^{\mathcal{I}}(\vec{S}_s)> U_{i}^{\mathcal{I}}(\vec{S}^*)$, where $\vec{S}_s=(s,\vec{S}^*_{-r_{i'}})$ is the strategy profile obtained by changing $i$'s strategy to $s$ in $\vec{S}^*$. 
Take $i'\in A'$ such that $\mu(i',\vec{S}^{*\prime})=i$ and let $\vec{S}'_{s}=(s\cup\{r_i\},\vec{S}^{*\prime}_{-i'})$ denote the strategy profile obtained by changing $i'$'s strategy to $s\cup\{r_i\}$ in $\vec{S}^{*\prime}$. Since both $\vec{S}^{*\prime}$ and $\vec{S}'_{s}$ are perfectly-matched, we obtain
$$U_{i'}^{\mathcal{I}'}(\vec{S}'_{s})=v'_{r_i}+U_{i}^{\mathcal{I}}(\vec{S}_{s})>v'_{r_i}+U_{i}^{\mathcal{I}}(\vec{S}^*)=U_{i'}^{\mathcal{I}'}(\vec{S}^{*\prime}),$$
which contradicts with the assumption that $\vec{S}^{*\prime}$ is a PNE of $\mathcal{I}$.
\end{proof}

With the help of technical Lemma \ref{lem:convert-to-symmetric-sandv}, we can modify Example \ref{exa:twoagent-noPNE} to obtain a counterexample showing that PNE may not exist even when there are only three agents, sharing a common strategy space, and only one agent is not weighted $1$.
\begin{theorem}
\label{thm:weighted-agent-noPNE}
There exists an instance of $(\vec{w})$-asymmetric CAG which has no pure Nash equilibrium, even if there are only three agents and only one agent is not weighted $1$.
\end{theorem}
\begin{proof}
    To prove this theorem, we only need to construct a ($\Vec{w}$)-asymmetric CAG which has no pure Nash equilibrium. Based on Example \ref{exa:twoagent-noPNE}, we can apply Lemma \ref{lem:convert-to-symmetric-sandv}. In this way, we can get the desirable instance. 
\end{proof}
The next question is, can we distinguish the instances of asymmetric CAG which have a PNE? We show that it is NP-complete to judge the existnece of PNE for any asymmetric CAG, as stated in Theorem \ref{thm:weighted-NP-hard}. 
\begin{theorem}\label{thm:weighted-NP-hard}
The decision problem asking whether a PNE exists for an instance $\mathcal{I}$ of $(\vec{w})$-asymmetric CAG is NP-complete, even if only one agent is not weighted $1$. As a corollary, deciding the existence of a PNE for an instance under $(\vec{\mathcal{S}},\vec{w},\vec{v})$-asymmetric CAG model is also NP-complete.
\end{theorem}
\begin{proof}
Firstly, it is easy to see that determining whether a PNE exists for any instance $\mathcal{I}$ under $(\vec{\mathcal{S}},\vec{w},\vec{v})$-asymmetric CAG model is in NP. 
We prove the NP-hardness by reduction from 3D-MATCHING problem. 

An instance of 3D-MATCHING is a tuple $(X,Y,Z,E)$, which consists of three disjoint vertex sets $X,Y,Z$ such that $|X|=|Y|=|Z|$, and a set of hyperedges $E\subseteq X\times Y\times Z$. The problem is to decide whether a perfect 3D-matching $M\subseteq E$ exists, with each vertex in $X,Y,Z$ appearing exactly once in $M$. It is well-known that 3D-MATCHING is NP-complete.

The reduction makes use of the instance $\mathcal{I}=(N,A,\vec{\mathcal{S}},\vec{w},\vec{v})$ in \Cref{exa:twoagent-noPNE}, which has no PNE. Recall that $N=\{q_1,q_2,q_3,q_4\}$ has four nodes with values $v_{q_1}=2,v_{q_2}=1,v_{q_3}=1,v_{q_4}=2$ and $A=\{1,2,3\}$ includes two active agents $1,2$ and one dummy agent $3$, weighted as $w_1=4,w_2=1,w_3=1$. We remove the dummy agent $3$ from $\mathcal{I}$, and denote the resulting instance as $\mathcal{I}'=(N,\{1,2\},(\mathcal{S}_1,\mathcal{S}_2),\vec{w}',\vec{v})$, where $w'_1=4$ and $w'_2=1$.

Observe that $\mathcal{I}'$ has two PNEs $(s_{1,1},s_{2,2})$ and $(s_{1,2},s_{2,1})$, which can be observed from the payoff matrix presented in \Cref{tab:modified-payoff} (Actually, the existence of PNE is also implied by \Cref{thm:exist asym two agent}, since there are only two agents in $\mathcal{I}'$). This indicates that $\mathcal{I}'$ can serve as a gadget converting the introduction of an extra agent on node $q_1,q_4$ (i.e., the dummy agent 3) into the non-existence of a PNE.

\begin{table}[htbp]
\centering
\begin{tabular}{|c|c|c|}
\hline
\diagbox{$S_1$}{$U_1,U_2$}{$S_2$}& $s_{2,1}$ & $s_{2,2}$\\\hline
$s_{1,1}$ &
(2.6,1.4) & (2.8,2.2) \\\hline
$s_{1,2}$ &
(2.8,2.2) & (2.6,1.4) \\\hline
\end{tabular}
\caption{The payoff matrix after removing the dummy agent}
\label{tab:modified-payoff}
\end{table}
Given an instance $(X,Y,Z,E)$ of 3D-MATCHING, suppose $|X|=|Y|=|Z|=n_V$. We first construct an instance $\mathcal{I}''$ of $(\vec{\mathcal{S}},\vec{w},\vec{v})$-asymmetric CAG using $\mathcal{I}'$ as a gadget, and then convert it to an instance $\mathcal{I}^*$ of $(\vec{w})$-asymmetric CAG by \Cref{lem:convert-to-symmetric-sandv}, such that for both $\mathcal{I}''$ and $\mathcal{I}^*$, a PNE exists if and only if a perfect 3D-matching exists. $\mathcal{I}''=(N'',A'',\vec{\mathcal{S}}'',\vec{w}'',\vec{v}'')$ is constructed as follows:
\begin{itemize}
    \item The node set is defined as $N''=N\cup \{q^V_j:j\in X\cup Y\cup Z\} \cup \{q^{F}_j:j=1,\cdots,n_V\}$.
    \item The agent set is defined as $A''=\{1,2\}\cup A^{match}$, where $A^{match}=\{3+1,\cdots,3+n_V\}$.
    \item Three new strategy spaces are defined as followings: $\mathcal{S}^{E}=\{\{q^V_x,q^V_y,q^V_z\}:(x,y,z)\in E\}$, $\mathcal{S}^{fail}=\{\{q^{F}_1,q_1,q_4\}\}\cup\{\{q^{F}_j\}:j=2,\cdots,n_V\}$ and $\mathcal{S}^{match}=\mathcal{S}^{E}\cup \mathcal{S}^{fail}$. For agents 1 and 2, keep the original strategy spaces, i.e., $\mathcal{S}_1''=\mathcal{S}_1$ and $\mathcal{S}_2''=\mathcal{S}_2$; For the new agent $i\in A^{match}$, let $\mathcal{S}_i''=\mathcal{S}^{match}$. 
    \item The weight of agent $1$ is still $w_1''=4$ and for all $i\in \{2\}\cup A^{match}$, they all have unit weight, $w_i''=1$.
    \item Keep $v''_{q_j}=v_{q_j}$ for each $q_j\in N=\{q_1,q_2,q_3,q_4\}$. Let $v_{q^V_j}''=10$ for all $q^V_j$ and $v_{q^{F}_j}''=26$ for all $q^{F}_j$.
\end{itemize}

Now we prove that $\mathcal{I}''$ has a PNE if and only if a perfect 3D-matching exists.
Intuitively, each node $q_{j}^{V}$ represents a vertex in  $X,Y,Z$, while each agent in $A^{match}$ tries to choose an hyperedge from $E$, represented by a strategy in $\mathcal{S}^{E}$. If a perfect 3D-matching exists, there is a strategy profile such that no two agents in $A^{match}$ cover a same node, which implies the existence of PNE. In contrast, if there is no perfect 3D-matching, there always exist two agents whose chosen hyperedges are overlapping. In this case, either of them can deviate to some strategy in $\mathcal{S}^{fail}$ to get higher utility. In consequence, $q_1,q_4$ will be attracted by an agent in $A^{match}$, which changes the payoff matrix between agents 1 and 2  in \Cref{tab:modified-payoff} into the one in \Cref{tab:nonexistence-PNE-payoff} and eliminates the existence of PNE.

\noindent\textbf{Sufficiency:}

Suppose there is a perfect 3D-matching $M\subseteq E$, i.e.,  $M=\{(x_1,y_1,z_1),\cdots,(x_{n_V},y_{n_V},z_{n_V})\}$. We show that a PNE $\vec{S}\in\vec{\mathcal{S}}''$ exists:
\begin{itemize}
    \item For agent 1 and 2, let $S_1=\{q_1,q_2\}$ and $S_2=\{q_2,q_4\}$. As previously discussed, both agent $1$ and $2$ have no incentive to deviate.
    \item For each agent $i\in A^{match}=\{3+1,\cdots,3+n_V\}$, let $S_i=\{q^V_{x_{i-3}},q^V_{y_{i-3}},q^V_{z_{i-3}}\}$. The utility of agent $i$ is $U_i^{\mathcal{I}''}(\vec{S})=30$. Note that the total value of any strategy in $\mathcal{S}^E$ is $30$ and the total value of any strategy in $\mathcal{S}^{fail}$ is no more than $v_{q^{F}_1}+v_{q_1}+v_{q_4}=26+2+2=30$. Therefore, no agent in $A^{match}$ wants to deviate.
\end{itemize}
Consequently, we have found that $\vec{S}$ is a PNE.

\noindent\textbf{Necessity:}

Suppose that there is a PNE $\vec{S}$ of instance $\mathcal{I}''$. We present some properties of $\vec{S}$. Firstly, every node in $\{q_j^V:j\in X\cup Y\cup Z\}$ is attracted by at most one agent in $A^{match}$. Otherwise, we can construct a contradiction: suppose that an agent $i\in A^{match}$ selects a strategy $S_i=\{q^V_x,q^V_y,q^V_z\}$ and at least one of $q^V_x,q^V_y,q^V_z$ is attracted by some other agents. Then, we have $U_i^{\mathcal{I}''}(\vec{S})\leq 10+10+10/2=25<26$. On the other hand, there are at least one node in $\{q_j^F:j=1,\cdots,n_V\}$ that is not attracted by any other agent, so agent $i$ can deviate to attract it and get a utility of at least $26$, which contradicts with that $\vec{S}$ is a PNE.

Secondly, every node in $\{q_j^{F}:j=1,\cdots,n_V\}$ is also attracted by at most one agent in $A^{match}$, because if an agent $i\in A^{match}$ is attracting some $q_j^{F}$ such that $q_j^{F}$ is also attracted by some other agents, then we have $U_i^{\mathcal{I}''}(\vec{S})\leq 26/2+2+2=17<26$. It means that agent $i$ can deviate to attract some $q_{j'}^{F}$ which is not attracted by any other agent, to improve her utility.

Thirdly, if there exist any agent $i\in A^{match}$ selecting strategies in $\mathcal{S}^{fail}$, then the node $q_1^F$ must be attracted by exactly one agent. We already know that $q_1^F$ cannot be attracted by more than one agents, and if $q_1^F$ is not attracted by any agent, using $\vec{S}_i = \{q_1^F,q_1,q_4\}$ will bring a utility more than 26, strictly better than any other strategy in $\mathcal{S}^{fail}$, which also leads to contradiction.

Now we know that, either all agents $i\in A^{match}$ select $S_i\in\mathcal{S}^E$, or there is exactly one agent $\hat{i}\in A^{match}$ selecting $S_{\hat{i}}=\{q_1^F,q_1,q_4\}$. In the latter case, since $\sum_{i\in A\setminus\{1,2\}:q_2\in S_i}w_i=1$, the two-player game between agent $1$ and $2$ has no PNE as discussed in \Cref{exa:twoagent-noPNE}. In other words, one of agent $1$ and $2$ has incentive to deviate from $\vec{S}$, contradicting with that $\vec{S}$ is a PNE. Therefore, only the former case can happen. In the former case, since every node in $\{q_j^V:j\in X\cup Y\cup Z\}$ is attracted by at most one agent in $A^{match}$, combining with that $|A^{match}|=n_V$, it is clear that $M=\{(x_i,y_i,z_i):S_{i+3}=\{q^V_{x_{i}},q^V_{y_{i}},q^V_{z_{i}}\}, \forall i=1,\cdots,n_V\}$ is a perfect 3D-matching.

In conclusion, $\mathcal{I}''$ has a PNE if and only if a perfect 3D-matching for the instance $(X,Y,Z,E)$ exists. Moreover, by \Cref{lem:convert-to-symmetric-sandv}, we can convert $\mathcal{I}''$ to an instance $\mathcal{I}^*$ of $(\vec{w})$-asymmetric CAG. This implies that $\mathcal{I}^*$ has a PNE if and only if a perfect 3D-matching for $(X,Y,Z,E)$ exists.  
\end{proof}

So far, we have shown that there exist instances which do not have a PNE and deciding whether the PNE exists is also NP-hard. However, if we focus on the scope of approximate PNE, we can make sure that a $O(log w_{\textnormal{max}})$-approximate PNE always exists for any $(\vec{\mathcal{S}},\vec{w},\vec{v})$-asymmetric CAG. 
\begin{theorem}
    For any instance of $(\vec{\mathcal{S}},\vec{w},\vec{v})$-asymmetric CAG, there exists a $O(\log w_{\textnormal{max}})$-approximate PNE, which also achieves the $O(\log W)$-approximately optimal social welfare, where $w_{\textnormal{max}}=\max_{i\in A}w_i$ is the maximum of agents' weights and $W=\sum_{i\in A}w_i$ is the sum of agents' weights.
\end{theorem}
\begin{proof}
We prove this theorem with the aid of a new potential function. Intuitively, when an agent deviates from her strategy to achieve at least $\alpha$ times her current utility, the value of the potential function will also be improved. After finite many deviations, we reach an $\alpha$-approximate PNE, since the joint strategy space is finite.

We define $\psi(x)=\ln(\max(e^{-1},x))$ and define the potential function as
$$\Psi(\vec{S})=\sum_{j\in N}v_j \psi(c(j,\vec{S})).$$ We first give a technical lemma:
\begin{lemma}
\label{lem:ln-approximate-potential}
For any $x\in \mathbb{Z}_{\geq 1}$, $y\in\mathbb{Z}_{\geq 0}$, it holds that
\begin{align*}
\frac{x}{x+y}\leq \psi(y+x)-\psi(y)\leq(\ln(1+x)+1)\frac{x}{x+y}.
\end{align*}
\end{lemma}
\begin{proof}
If $y=0$, the inequalities hold trivially, because 
$$\frac{x}{x}=1\leq \ln(x)-\ln(e^{-1})=\ln(x)+1\leq(\ln(1+x)+1)\frac{x}{x}.$$
If $y\geq 1$, by the definition of function $\psi(\cdot)$, we know $\psi(y+x)-\psi(y)=\ln(y+x)-\ln(y)$. For the first inequality, we have
$$\frac{x}{y+x} \leq \int_{y}^{y+x}\frac1tdt = \ln(y+x)-\ln(y)= \psi(y+x)-\psi(y),$$ which implies that the first inequality holds.
Then, we show that the second inequality holds. By the following a series of equation and inequalities
$$(1+\frac{x}{y})^{\frac{x+y}{x}}=(1+\frac{x}{y})^{\frac{y}{x}}(1+\frac{x}{y})\leq e(1+\frac{x}{y})\leq e(1+x),$$ 
and taking logarithm on both sides, we get $$\frac{x+y}{x}\ln(1+\frac{x}{y})\leq \ln(1+x)+1.$$ 
Consequently, we obtain
\begin{align*}
\psi(y+x)-\psi(y)=\ln(y+x)-\ln(y)=\ln(1+\frac{x}{y})\leq (\ln(1+x)+1)\frac{x}{x+y}.
\end{align*}
\end{proof}

Let $\alpha=\ln(1+w_{\textnormal{max}})+1=\Theta(\log w_{\textnormal{max}})$. We show that there is an $\alpha$-approximate PNE which achieves the $O(\log W)$-approximately optimal social welfare.

For any strategy profile $\vec{S}\in\vec{\mathcal{S}}$, if $\vec{S}$ is not an $\alpha$-approximate PNE, then by definition there is $i\in A$, such that $i$ can deviate to $S_i'\in\mathcal{S}_i$ and improve her utility with an $\alpha$ multiplicative factor. This implies that
\begin{align*}
&\alpha U_i(S_i,\vec{S}_{-i})< U_i(S_i',\vec{S}_{-i})\\ \Rightarrow & \alpha(\sum_{j\in S_i}v_j\frac{w_i}{w_i+c(j,\vec{S}_{-i})})-\sum_{j\in S_i'}v_j\frac{w_i}{w_i+c(j,\vec{S}_{-i})}< 0\\
 \Rightarrow & \alpha(\sum_{j\in S_i\setminus S_i'}v_j\frac{w_i}{w_i+c(j,\vec{S}_{-i})})-\sum_{j\in S_i'\setminus S_i}v_j\frac{w_i}{w_i+c(j,\vec{S}_{-i})}\nonumber\\
&< (1-\alpha)\sum_{j\in S_i\cap S_i'}v_j\frac{w_i}{w_i+c(j,\vec{S}_{-i})}\nonumber \leq 0.
\end{align*}
By \Cref{lem:ln-approximate-potential}, for any $j\in N$, we have 
\begin{align*}
\frac{w_i}{w_i+c(j,\vec{S}_{-i})} \leq \psi(c(j,\vec{S}_{-i})+w_i)-\psi(c(j,\vec{S}_{-i})) \\
\leq (\ln(1+w_i)+1)\frac{w_i}{w_i+c(j,\vec{S}_{-i})} \leq \alpha\frac{w_i}{w_i+c(j,\vec{S}_{-i})}.
\end{align*}
Then, we calculate the difference in the potential function value:
\begin{align*}
&\Psi(S_i',\vec{S}_{-i})-\Psi(S_i,\vec{S}_{-i})\\
=&\sum_{j\in S_i'\setminus S_i}v_j(\psi(c(j,\vec{S}_{-i})+w_i)-\psi(c(j,\vec{S}_{-i})))\nonumber\\
- & \sum_{j\in S_i\setminus S_i'}v_j(\psi(c(j,\vec{S}_{-i})+w_i)-\psi(c(j,\vec{S}_{-i})))\\
\geq& \sum_{j\in S_i'\setminus S_i}v_j\frac{w_i}{w_i+c(j,\vec{S}_{-i})}-\alpha(\sum_{j\in S_i\setminus S_i'}v_j\frac{w_i}{w_i+c(j,\vec{S}_{-i})}) >0.
\end{align*}
Therefore, the value of $\Psi$ increases after every such $\alpha$-improving deviation. Since the set of possible strategy profiles is finite, one can find an $\alpha$-approximate PNE after finite steps of $\alpha$-improving deviation.

Specially, there exists a strategy profile maximizing $\Psi(\vec{S})$. We formally define it as $\vec{S}^{\Psi*}\in\arg\max_{\vec{S}\in\vec{\mathcal{S}}}$ $\Psi(\vec{S})$ and know that there is no $\alpha$-improving deviation from $\vec{S}^{\Psi*}$, which means that $\vec{S}^{\Psi*}$ is an $\alpha$-approximate PNE. Now we show that $\vec{S}^{\Psi*}$ achieves the $O(\log W)$-approximately optimal social welfare.

For any $x\in\mathbb{Z}_{\geq 0}$, it's easy to check that $I_{>0}(x)\leq \psi(x)+1\leq (\ln(x)+1)I_{>0}(x)$, where $I_{>0}(x)$ is the indicator function of positive integers. 
Let $\vec{S}^*$ be the strategy profile with the optimal social welfare, then 
\begin{align*}
SW(\vec{S}^*)=\sum_{j\in N}v_jI_{>0}(c(j,\vec{S}^*))\leq \sum_{j\in N}v_j(\psi(c(j,\vec{S}^*))+1) .
\end{align*}
On the other hand, for the strategy profile $\vec{S}^{\Psi*}$, we have 
\begin{align*}
    &\sum_{j\in N}v_j(\psi(c(j,\vec{S}^{\Psi*}))+1) \leq 
    \sum_{j\in N}[v_j\ln(c(j,\vec{S}^{\Psi*})+1)  \cdot I_{>0}(c(j,\vec{S}^{\Psi*}))]\\  & \leq  \ln(W+1)SW(\vec{S}^{\Psi*}).
\end{align*}
Due to  
$$\sum_{j\in N}v_j(\psi(c(j,\vec{S}^*))+1)=\Psi(\vec{S}^*)+\sum_{j\in N}v_j\leq\Psi(\vec{S}^{\Psi*})+\sum_{j\in N}v_j=\sum_{j\in N}v_j(\psi(c(j,\vec{S}^{\Psi*}))+1),$$ we can combine the above two inequalities and obtain  $SW(\vec{S}^*) \leq \ln(W+1)SW(\vec{S}^{\Psi*})$.   
\end{proof}

As for the lower bound of approximation ratio of approximate PNE, although we know that there are some instances that do not exist PNE (even does not exist 10/9-approximate PNE from Example \ref{exa:twoagent-noPNE}), we still hope to find a lower bound as large as possible. In the following theorem, we technically construct an instance which does not exist a $(2-\epsilon)$-approximate PNE.
\begin{theorem}
    For any $0<\epsilon<1$, we can construct an instance $\mathcal{I}=(N,A,\Vec{\mathcal{S}},\Vec{w},\Vec{v})$ of $(\Vec{\mathcal{S}},\Vec{w},\Vec{v})$-asymmetric CAG which does not admit a $(2-\epsilon)$-approximate PNE.
\end{theorem}
\begin{proof}
    In this instance, there are three groups of agents. The first group has $k$ agents $1,2,...k$, and their weights are $p,p^2,...,p^k$ respectively. The second group has $l$ agents $k+1,k+2,...,k+l$, and their weights are all 1. The last group has 2 agents $k+l+1,k+l+2$, and their weights are both $p^k$. 

    On the other hand, there are $2k+2l$ nodes labeled as $N=\{q_1,\overline{q}_1,q_2,\overline{q}_2,...,q_{k+l},\overline{q}_{k+l}\}$. Let $v_{q_i}=v_{\overline{q}_i}=x_i$ for $i=1,2,...,k$ and let $v_{q_i}=v_{\overline{q}_i}=x_0$ for $i=k+1,k+2,...,k+l$.
    
    The attraction ranges and strategy spaces can be seen in Figure \ref{fig:fig1}. For $i=2,3,...,k$, the strategy space of the agent $i$ is $\mathcal{S}_i=\{s_{i,1},s_{i,2}\}$, where $s_{i,1}=\{q_{i-1},q_i\}$ and $s_{i,2}=\{\overline{q}_{i-1},\overline{q}_i\}$. For agent $1$, the strategy space is $\mathcal{S}_1=\{s_{1,1},s_{1,2}\}$, where $s_{1,1}=\{q_1,q_{k+1},q_{k+2},...,q_{k+l}\}, s_{1,2}=\{\overline{q}_1,\overline{q}_{k+1},\overline{q}_{k+2},...,\overline{q}_{k+l}\}$. For $i=k+1,k+2,...,k+l$, the strategy space of the agent $i$ is $\mathcal{S}_i=\{s_{i,1},s_{i,2}\}$, where $s_{i,1}=\{q_i,\overline{q}_k\}, s_{i,2}=\{\overline{q}_i,q_k\}$. The last two are dummy agents, and $\mathcal{S}_{k+l+1}={q_{k+1}},\mathcal{S}_{k+l+2}={\overline{q}_{k+1}}$.

    Now we start to prove that the instance identified above does not admit a $(2-\epsilon)$-approximate PNE. As shown in Figure \ref{fig:fig1}, we provide the intuition first. When agent $k$ selects $s_{k,1}=\{q_{k-1},q_k\}$, agent $k-1$ will select $s_{k-1,2}=\{\overline{q}_{k-2},\overline{q}_{k-1}\}$ to obtain the full value of $\overline{q}_{k-1}$. Then agent $k-2$ will select $s_{k-2,1}=\{q_{k-3},q_{k-2}\}$. Keep selecting like this, agent $2$ will select $s_{2,2}=\{\overline{q}_{2},\overline{q}_{1}\}$ if $k$ is an odd number. Then agent $1$ will select $s_{1,1}=\{q_1,q_{k+1},q_{k+2},...,q_{k+l}\}$ to obtain the full value of $q_1$. Next, $q_i$ will select $s_{i,2}=\{\overline{q}_i,q_k\}$ for $i=k+1, k+2, ..., k+l$. At this time, agent $k$ will select $s_{k,2}=\{\overline{q}_{k-1},\overline{q}_{k}\}$ since $q_k$ is selected by agent $k+1, k+2, ..., k+l$. Similarly, all other agents will change their selection one by one, which suggests that there is no $(2-\epsilon)$-approximate PNE in the game.

     \begin{figure}[h]
		\begin{center}
			\begin{tabular}{c}						\includegraphics[width=12cm,height=4.5cm]{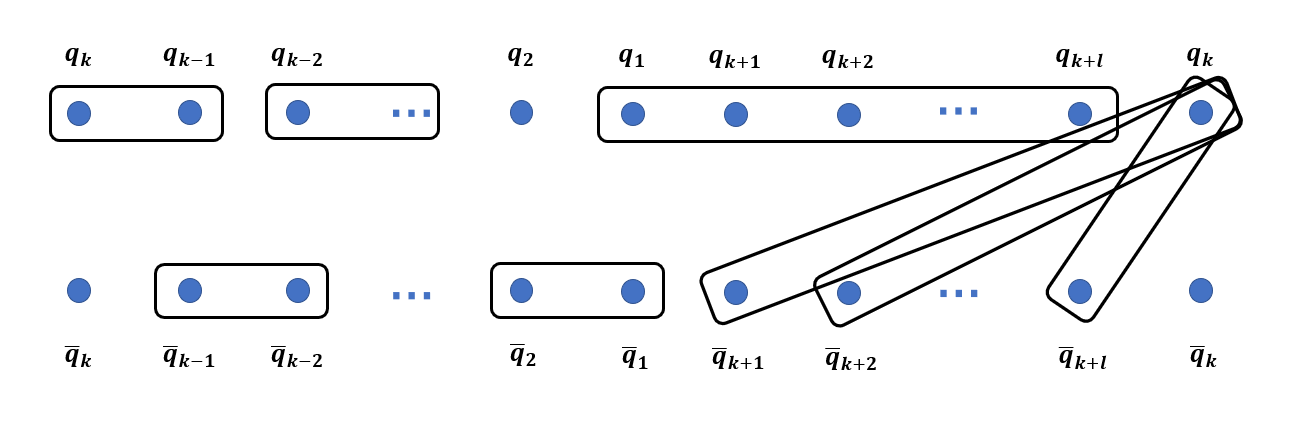}
			\end{tabular}
			\end{center}
			\caption 
			{The strategies that agents will select in the game.} 
                \label{fig:fig1}
			\end{figure} 

    Firstly, for $i=3,4, ...,k$, if agent $i$ selects $s_{i,1}=\{q_{i-1},q_i\}$, agent $i-1$ will select $s_{i-1,2}=\{\overline{q}_{i-2},\overline{q}_{i-1}\}$ no matter what others choose. Otherwise, if agent $i-1$ selects $s_{i-1,1}=\{q_{i-2},q_{i-1}\}$, the maximum utility is $x_{i-2}+x_{i-1}\cdot\frac{p^{i-1}}{p^{i-1}+p^i}$. On the other hand, the minimum utility of agent $i-1$ to select $s_{i-1,2}$ is $x_{i-2}\cdot\frac{p^{i-1}}{p^{i-1}+p^{i-2}}+x_{i-1}$. Therefore, we only need
    \begin{align} \label{eqn1}
       (2-\epsilon) \cdot (x_{i-2}+x_{i-1}\cdot\frac{p^{i-1}}{p^{i-1}+p^i})\leq x_{i-2}\cdot\frac{p^{i-1}}{p^{i-1}+p^{i-2}}+x_{i-1}
    \end{align}
    to ensure the deviation from $s_{i-1,1}$ to $s_{i-1,2}$ is $(2-\epsilon)$-improving.

    Next, if agent $2$ selects $s_{2,2}=\{\overline{q}_{2},\overline{q}_{1}\}$, agent $1$ will select $s_{1,1}=\{q_1,q_{k+1},q_{k+2},...,q_{k+l}\}$ no matter what others choose. Otherwise, if agent $1$ selects $s_{1,2}=\{\overline{q}_1,\overline{q}_{k+1},\overline{q}_{k+2},...,\overline{q}_{k+l}\}$, the maximum utility is $x_1\cdot\frac{p}{p^2+p}+l\cdot x_0$. On the other hand, the minimum utility of agent $1$ to select $s_{1,1}$ is $x_1+l\cdot\frac{p}{p+1}\cdot x_0$. Therefore, we only need 
     \begin{align} \label{eqn2}
       (2-\epsilon) \cdot (x_1\cdot\frac{p}{p^2+p}+l\cdot x_0)\leq x_1+l\cdot\frac{p}{p+1}\cdot x_0
    \end{align}
    to ensure the deviation from $s_{1,2}$ to $s_{1,1}$ is $(2-\epsilon)$-improving.

    Next, for $i=k+1, k+2, ..., k+l$, if agent $1$ selects $s_{1,1}=\{q_1,q_{k+1},q_{k+2},...,q_{k+l}\}$, agent $i$ will select $s_{i,2}=\{\overline{q}_i,q_k\}$ no matter what others choose. Otherwise, if agent $i$ selects $s_{i,1}=\{q_i,\overline{q}_k\}$, the maximum utility is $x_0\cdot\frac{1}{p+1}+x_k\cdot\frac{1}{p^k+1}$. On the other hand, the minimum utility of agent $i$ to select $s_{i,2}$ is $x_0+x_k\cdot\frac{1}{2p^k+l}$. Therefore, we only need
    \begin{align} \label{eqn3}
        (2-\epsilon) \cdot (x_0\cdot\frac{1}{p+1}+x_k\cdot\frac{1}{p^k+1}) \leq x_0+x_k\cdot\frac{1}{2p^k+l}
    \end{align}
    to ensure the deviation from $s_{i,1}$ to $s_{i,2}$ is $(2-\epsilon)$-improving.

    Ultimately, if agent $i$ selects $s_{i,1}=\{q_i,\overline{q}_k\}$ for each $i=k+1, k+2, ..., k+l$, agent $k$ will select $s_{k,1}=\{q_k,q_{k-1}\}$ no matter what others choose. Otherwise, if agent $k$ selects $s_{k,2}=\{\overline{q}_k,\overline{q}_{k-1}\}$, the maximum utility is $x_k\cdot\frac{p^k}{2p^k+l}+x_{k-1}$. On the other hand, the maximum utility of agent $k$ to select $s_{k,1}=\{q_k,q_{k-1}\}$ is $x_k\cdot\frac{p^k}{2p^k}+x_{k-1}\frac{p^k}{p^k+p^{k-1}}$. Therefore, we only need
    \begin{align} \label{eqn4}
        (2-\epsilon) \cdot (x_k\cdot\frac{p^k}{2p^k+l}+x_{k-1})\leq x_k\cdot\frac{p^k}{2p^k}+x_{k-1}\frac{p^k}{p^k+p^{k-1}}
    \end{align}
    to ensure the deviation from $s_{k,2}$ to $s_{k,1}$ is $(2-\epsilon)$-improving.

    Combining Inequality \ref{eqn1}-\ref{eqn4}, we have
    \begin{align} \label{eqn5}
        x_{i-2}\leq\frac{p-1+\epsilon}{(1-\epsilon)p+2-\epsilon}\cdot x_{i-1},\quad i=3, 4, ..., k
    \end{align}
    \begin{align} \label{eqn6}
        x_0\leq\frac{p-1+\epsilon}{l[(1-\epsilon)p+2-\epsilon]}\cdot x_1,
    \end{align}
    \begin{align} \label{eqn7}
        x_k \leq \frac{(p-1+\epsilon)(p^k+1)(2p^k+l)}{(p+1)[(2-\epsilon)(2p^k+l)-p^k-1]}\cdot x_0,
    \end{align}
    \begin{align} \label{eqn8}
        x_{k-1}\leq\frac{(p+1)[(2\epsilon-2)p^k+l]}{2(2p^k+l)[(1-\epsilon)p+2-\epsilon]}\cdot x_{k}.
    \end{align}

    Then we set $x_0, x_1, ..., x_k$ to let Inequality \ref{eqn5}, \ref{eqn6}, \ref{eqn7} equal. To prove that the instance does not admit a $(2-\epsilon)$-approximate PNE, we only need to prove
    \begin{align} \label{eqn9}
        [\frac{p-1+\epsilon}{(1-\epsilon)p+2-\epsilon}]^k\cdot\frac{(p^k+1)[(2\epsilon-2)p^k+l]}{2l[(2-\epsilon)(2p^k+l)-p^k-1]}\geq1
    \end{align}
    to ensure there is no $(2-\epsilon)$-approximate PNE in the game. Specifically, we let $p \geq \epsilon+\frac{6}{\epsilon}-5$, let $k\geq \log_{\frac{1}{1-\frac{\epsilon}{2}}}88$ and $l=4p^k$. Then the left side of Inequality \ref{eqn9} is greater than
    \begin{align*}
        (\frac{1}{1-\frac{\epsilon}{2}})^k\cdot\frac{(\epsilon+1)(p^k+1)}{4[(11-6\epsilon)p^k-1]}\geq1.
    \end{align*}
    Therefore, the instance we constructed does not admit a $(2-\epsilon)$-approximate PNE.
\end{proof}

Lastly, we discuss the PoA of asymmetric CAG. One can observe that the proof of upper bound $2$ on PoA in \Cref{lem:up_poa} actually doesn't require any symmetry of agents and nodes. Therefore, as a corollary, over all instances of asymmetric CAG admitting PNEs, we have $\text{PoA}=2$. Nevertheless, we also know that approximate PNE always exists for suitable approximation ratios. If we consider the PoA defined on the range of $\alpha$-approximate PNE, we can get a general conclusion as an extension of Theorem \ref{lem:up_poa}. 
\begin{theorem}\label{thm:general PoA}
    For any instance of $(\vec{\mathcal{S}},\vec{w},\vec{v})$-asymmetric CAG, the PoA for $\alpha$-approximate PNE is upper bounded by $1+\alpha$.
\end{theorem}
\begin{proof}
Suppose that $\vec{S}^{\alpha}$ is an $\alpha$-approximate PNE for some $\alpha\geq 1$ and $\vec{S}^*$ is the strategy profile with the optimal social welfare. By the definition of $\alpha$-approximate PNE, for any agent $i\in A$, we have 
\begin{align*}
U_i(S^*_i,\vec{S}^{\alpha}_{-i})\leq \alpha U_i(\vec{S}^{\alpha}).
\end{align*}

We construct a strategy profile as a medium to show this theorem. Assume that each agent $i$ uses the strategy $S^*_i\cup S^{\alpha}_i$ (although it is possibly not in $\mathcal{S}_i$) and define the medium strategy profile as $\vec{S}^{+}=(S^*_1\cup S^{\alpha}_1,\cdots,S^*_m\cup S^{\alpha}_m)$. By the definition of social welfare, i.e., $SW(\vec{S})=\sum_{j\in \bigcup_{i=1}^m S_i}v_j$ and $\bigcup_{i=1}^m S^*_i\subseteq \bigcup_{i=1}^m (S^*_i\cup S^{\alpha}_i)$, it follows that 
\begin{align*}
SW(\vec{S}^*)\leq \sum_{j\in \bigcup_{i=1}^m S^*_i}v_j \leq \sum_{j\in \bigcup_{i=1}^m (S^*_i\cup S^{\alpha}_i)}v_j =SW(\vec{S}^{+}).
\end{align*}

On the other hand, \begin{align}
SW(\vec{S}^{+})&=\sum_{i=1}^m U_i(S^*_i\cup S^{\alpha}_i,\vec{S}^{+}_{-i})\nonumber\\
&\leq \sum_{i=1}^m U_i(S^*_i\cup S^{\alpha}_i,\vec{S}^{\alpha}_{-i})\label{alpha-appx ineq 1}\\
&=\sum_{i=1}^m U_i(S^*_i,\vec{S}^{\alpha}_{-i})+
\sum_{i=1}^m U_i(S^{\alpha}_i \backslash S^*_i,\vec{S}^{\alpha}_{-i})\nonumber \\
&\leq \sum_{i=1}^m U_i(S^*_i,\vec{S}^{\alpha}_{-i})+
\sum_{i=1}^m U_i(S^{\alpha}_i,\vec{S}^{\alpha}_{-i})\nonumber \\
&\leq \sum_{i=1}^m(1+\alpha)U_i(\vec{S}^{\alpha})\label{alpha-appx ineq 2}\\
&=(1+\alpha)SW(\vec{S}^{\alpha}).\nonumber
\end{align}
The inequality (\ref{alpha-appx ineq 1}) is due to that when the strategies of agents except $i$ do not increase, the utility of agent $i$ will not decrease. The inequality (\ref{alpha-appx ineq 2}) is based on the definition of $\alpha$-approximate PNE. Thus, we have $SW(\vec{S}^*)\leq SW(\vec{S}^+)\leq (1+\alpha)SW(\vec{S}^{\alpha})$. 
\end{proof}







\section{The Sequential Game}
\label{Sec:sequential}

In reality, 
the agents may take actions sequentially in most of the situations.  
In these scenarios, the agents who act later may be advantaged, since they can see the actions of the former ones and optimize their own strategies. In this section, we study the CAG in sequential setting. 

In the sequential CAG, an instance is still defined as the tuple $(N,A,\vec{\mathcal{S}},\vec{w},\vec{v})$. Without loss of generality, we assume that the agents are labeled as $A=\{1,\cdots,m\}$ and move in the order of $1,\cdots,m$. In round $i\in\{1, 2, \cdots, m\}$, the agent $i$ observes the strategies selected by agents $1$ to $i-1$, denoted by $\vec{S}_{< i}=(S_1,\cdots,S_{i-1})\in \times_{i'=1}^{i-1}\mathcal{S}_{i'}$ and decides its own strategy $S_i=\sigma_i(\vec{S}_{<i})$. We call $\sigma_i:\times_{i'=1}^{i-1}\mathcal{S}_{i'}\to\mathcal{S}_i$ as the strategy list of agent $i$. Given the strategy lists profile $\vec{\sigma}=(\sigma_1,\cdots,\sigma_m)$, the outcome of $\vec{\sigma}$ is the strategy profile formed in the end, formally defined as $\vec{\alpha}(\vec{\sigma})=(\alpha_1(\vec{\sigma}),\cdots,\alpha_m(\vec{\sigma}))\in\vec{\mathcal{S}}$, such that $\alpha_1(\vec{\sigma})=\sigma_1(\emptyset)$ and $\alpha_i(\vec{\sigma})=\sigma_i(\alpha_1(\vec{\sigma}),\cdots,\alpha_{i-1}(\vec{\sigma}))$ for any $i \in \{2, \cdots, m\}$.

When any prefix $\vec{S}_{<i}$ is fixed for some $i\in A$, a subgame of agents $i,\cdots, m$ is induced naturally, and the outcome in this subgame is denoted by $\vec{\alpha}(\vec{S}_{<i},\sigma_{i}, \cdots, \sigma_m)\in\vec{\mathcal{S}}$. We introduce the subgame perfect equilibrium (SPE), defined as follows.
\begin{definition}
A strategy lists profile $\vec{\sigma}$ is a subgame perfect equilibrium, if and only if for any agent $i\in A$, any prefix $\vec{S}_{<i}\in \times_{i'=1}^{i-1}\mathcal{S}_{i'}$ and any $S'_i\in\mathcal{S}_i$, it holds that
$$U_i\left(\vec{\alpha}(\vec{S}_{<i},\sigma_{i}, \cdots, \sigma_m)\right)\geq U_i\left(\vec{\alpha}\left((\vec{S}_{<i}, S'_{i}), \sigma_{i+1}, \cdots, \sigma_m\right)\right).$$
Let $\text{SPE}(\mathcal{I})$ denote the set of all subgame perfect equilibria of instance $\mathcal{I}$.
\end{definition}

It is well-known that the SPE always exists and can be found by backward induction: when $\sigma_{i+1},\cdots,\sigma_{m}$ is given, agent $i$ can select $\sigma(\vec{S}_{<i})\in\arg\max_{S_i\in\mathcal{S}_i}U_i$ $\left(\vec{\alpha}(\vec{S}_{<i},S_i,\sigma_{i+1}, \cdots, \sigma_m)\right)$ for each $\vec{S}_{<i}\in \times_{i'=1}^{i-1}\mathcal{S}_{i'}$.
Now we come to discuss the computational complexity problem of finding an SPE. Since the description length of an SPE is exponentially large, we only care about finding an SPE outcome. However, this is still computationally difficult, as we can prove that a decision version of it is PSPACE-hard. 
\begin{theorem}\label{thm:sequential pspace hard}
Given an instance $\mathcal{I}$ of sequential CAG, an agent $i\in A$ and an rational number $X$, the problem to decide whether there exists an SPE $\vec{\sigma}\in \emph{SPE}(\mathcal{I})$ such that $U_i(\vec{\alpha}(\vec{\sigma}))\geq X$ is PSPACE-hard.
\end{theorem}
\begin{proof}
We give a reduction from TQBF \cite{LS80}, which is a PSPACE-complete problem. The TQBF problem asks the value of a fully quantified boolean formula $$\Psi=Q_1x_1Q_2x_2\cdots Q_nx_n\psi(x_1,\cdots,x_n),$$
where each $Q_i$ is a quantifier $\exists$ or $\forall$, and each $x_i\in\{0,1\}$ is a boolean varaiable, and $\psi(x_1,\cdots,x_n)$ is a boolean formula. 

For simplicity, we use $\vec{x}$ to denote $(x_1,\cdots,x_n)$. Without loss of generality, we assume that $n=2n'+1$ for some $n'\in\mathbb{Z}_{\geq 1}$, that $Q_1=Q_3=\cdots=Q_n=\exists$, $Q_2=Q_4=\cdots=Q_{n-1}=\forall$, and that $\psi(\vec{x})$ is written as a 3-CNF $\wedge_{i=1}^{n^C}(y_{i,1}\vee y_{i,2} \vee y_{i,3})$, where each $C_i(\vec{x})=(y_{i,1}\vee y_{i,2} \vee y_{i,3})$ is a clause, and each $y_{i,j}$ is a literal in the form of $x_{t_{i,j}}$ or $\neg x_{t_{i,j}}$.

We construct an instance $\mathcal{I}=(N,A,\vec{\mathcal{S}},\vec{w},\vec{v})$ of sequential CAG  as follows:
\begin{itemize}
    \item Define the node set $N=\{q^{\exists},q^{\forall}\}\cup N^{clause}\cup N^{literal}\cup N^{dummy}$, where $N^{clause}=\{q^{C}_i:i=1,\cdots,n^C\}$ corresponds with the clauses, $N^{literal}=\{q^{x_1},q^{\neg x_1},\cdots,q^{x_n},q^{\neg x_n}\}$ corresponds with the literals, and $N^{dummy}=\{q^D_0,q^D_1,q^D_2,q^D_3\}$.
    \item Define the agent set $A=\{1,\cdots,m\}$ with $m=n+5$.
    \item Construct strategy spaces $\vec{\mathcal{S}}$ as follows:
\begin{itemize}
    \item For agent $i=1,3,\cdots,n$, define $\mathcal{S}_i=\{\{q^{\exists},q^{x_i}\},\{q^{\exists},q^{\neg x_i}\}\}$;
    \item For agent $i=2,4,\cdots,n-1$, define $\mathcal{S}_i=\{\{q^{\forall},q^{x_i}\},\{q^{\forall},q^{\neg x_i}\}\}$;
    \item For agent $n+1$, define $\mathcal{S}_{n+1}=\{s_{n+1,j}:j=1,\cdots,n^C\}$, where $s_{n+1,j}=\{q^{\forall},q^C_1,q^C_2,\cdots,q^C_{n^C}\}\setminus\{q^C_j\}$;
    \item For agent $n+2$, define $\mathcal{S}_{n+2}=\{s^{T}_{n+2,j,k}:j=1,\cdots,n^C,k=1,2,3\}\cup\{s^{F}_{n+2}\}$, where $s^{T}_{n+2,j,k}=\{q^{\forall},q^C_j,q^{\neg y_{j,k}}\}$, $s^{F}_{n+2}=\{q^{\exists},q^D_0,q^D_1,q^D_2,q^D_3\}$.
    \item For each agent $i=n+3,n+4,n+5$, let $i$ be a dummy agent with $\mathcal{S}_i=\{\{q^D_1,q^D_2,q^D_3\}\}$.
\end{itemize}
    \item All agents and nodes have unit weights and unit values respectively, i.e., $w_i=1$ for all $i\in A$ and $v_j=1$ for all $j\in N$.
\end{itemize}

Note that the construction of $\mathcal{I}$ is polynomial-time computable. Intuitively, each agent $i=1,\cdots,n$ chooses the value of the variable $x_i$, where attracting $q^{x_i}$ indicates setting $x_i=1$, and attracting $q^{\neg x_i}$ indicates setting $x_i=0$. The agent $n+1$ tries to pick a false clause $C_j$ by selecting the strategy $s_{n+1,j}$. If the picked clause $C_j$ is true, agent $n+2$ will choose some $s^{T}_{n+2,j,k}$ where the literal $y_{j,k}$ is true, otherwise it will choose $s^{F}_{n+2}$. We will see that agents $1,3,\cdots,n$ try to prevent agent $n+2$ from choosing $s^F_{n+2}$, while agents $2,4,\cdots,n-1$ and $n+1$ try to prevent agent $n+2$ from choosing any $s^T_{n+2,j,k}$.

Let $\vec{\sigma}$ be an SPE of $\mathcal{I}$.
We first discuss the equilibrium strategy list $\sigma_{n+2}$ of agent $n+2$. Given any prefix $\vec{S}_{<n+2}\in\times_{i=1}^{n+1}\mathcal{S}_{i}$, an assignment $\vec{x}\in\{0,1\}^n$ and a clause $C_j$ are induced. The suffix $\vec{S}_{n+3,\cdots,n+5}$ is trivially fixed as they are dummy players. We calculate the utility of agent $n+2$, given its strategy:

(a) If $S_{n+2}=s^{F}_{n+2}$ and $\vec{S}=(\vec{S}_{<n+2},s^{F}_{n+2},\vec{S}_{n+3,\cdots,n+5})$, we have  $c(q^{\exists},\vec{S})=n'+2$ and the utility of agent $n+2$ is   $U_{n+2}(\vec{S})=\frac1{n'+2}+1+\frac14+\frac14+\frac14=\frac74+\frac1{n'+2};$

(b) If $S_{n+2}=s^{T}_{n+2,j',k}$, where $j'\neq j$, 
we know $q^C_{j'}$ is also covered by agent $n+1$, and $c(q^{\forall},\vec{S})=n'+2$. The utility of agent $n+2$ is $U_{n+2}(\vec{S})\leq \frac1{n'+2}+\frac12+1=\frac32+\frac1{n'+2};$

(c.1) If $S_{n+2}=s^{T}_{n+2,j,k}$ and $y_{j,k}=1$, 
we have $c(q^C_j,\vec{S})=1$, $c(q^{\neg y_{j,k}},\vec{S})=1$, and $c(q^{\forall},\vec{S})=n'+2$. The utility of agent $n+2$ can be calculated as $U_{n+2}(\vec{S})\leq \frac1{n'+2}+1+1=2+\frac1{n'+2}$;

(c.2) If $S_{n+2}=s^{T}_{n+2,j,k}$ and $y_{j,k}=0$, 
we have $c(q^C_j,\vec{S})=1$, $c(q^{\neg y_{j,k}},\vec{S})=2$, and $c(q^{\forall},\vec{S})=n'+2$. In this case, The utility of agent $n+2$ can be presented as $U_{n+2}(\vec{S})\leq \frac1{n'+2}+1+\frac12=\frac32+\frac1{n'+2}$;

By comparing these cases, we can observe that agent $n+2$ gets strictly better utility in case (a) than case (b) and case (c.2), which implies that only case (a) or case (c.1) will happen in the SPE outcome. Furthermore, when $C_j(\vec{x})=0$, we have $y_{j,1}=y_{j,2}=y_{j,3}=0$, which suggests that case (c.1) is impossible in this condition. By the definition of SPE, agent $n+2$ goes for case (a) to maximize its utility and we have $\sigma_{n+2}(\vec{S}_{<n+2})=s^{F}_{n+2}$. When $C_j(\vec{x})=1$, it holds $y_{j,1}\vee y_{j,2}\vee y_{j,3}=1$, so case (c.1) is possible for agent $n+2$. Therefore, we have $\sigma_{n+2}(\vec{S}_{<n+2})=s^{T}_{n+2,j,k}$, for some $k\in\{1,2,3\}$, such that $y_{j,k}=1$.

Recall that $c(q^{\exists},\vec{S})=n'+2$ and $c(q^{\forall},\vec{S})=n'+1$ in case (a), while the corresponding values $c(q^{\exists},\vec{S})=n'+1$ and $c(q^{\forall},\vec{S})=n'+2$ in case (c.1). Then we calculate the utilities of other agents: 

\begin{itemize}
    \item For agent $i=1,3,\cdots,n$, $U_{i}(\vec{S})=\frac1{n'+2}+1$ in case (a) and $U_{i}(\vec{S})=\frac1{n'+1}+1$ in case (c.1);
    \item For agent $i=2,4,\cdots,n-1$, $U_{i}(\vec{S})=\frac1{n'+1}+1$ in case (a) and $U_{i}(\vec{S})=\frac1{n'+2}+1$ in case (c.1);
    \item For agent $n+1$, $U_{n+1}(\vec{S})=\frac1{n'+1}+n^C-1$ in case (a)  and $U_{n+1}(\vec{S})=\frac1{n'+2}+n^C-1$ in case (c.1).
\end{itemize}

Let $\Psi_i(x_1,\cdots,x_{i-1})$ denote $Q_ix_iQ_{i+1}x_{i+1}\cdots Q_nx_n\psi(x_1,\cdots,x_n)$ for $i=1,2,\cdots,n$, and $\Psi_{n+1}(x_1,\cdots,x_n)=\psi(x_1,\cdots,x_n)$. Now we use backward induction to prove that $U_1(\vec{S}(\vec{\sigma}))=\frac1{n'+1}+1$ if $\Psi=1$, and $U_1(\vec{S}(\vec{\sigma}))=\frac1{n'+2}+1$ if $\Psi=0$.

For agent $n+1$ and any prefix $\vec{S}_{<n+1}$, let $\vec{x}$ denote the induced assignment. If $\psi(\vec{x})=1$, for any choice of $\sigma_{n+1}(\vec{S})=s_{n+1,j}$, it holds that $C_j(\vec{x})=1$ and the agent $n+2$ will always choose case (c.1). If $\psi(\vec{x})=0$, we have $\sigma_{n+1}(\vec{S})=s_{n+1,j}$ for some $j$ such that $C_j(\vec{x})=0$, and agent $n+2$ chooses case (a).

For agents $i=n,n-1,\cdots,1$ and any prefix $\vec{S}_{<i}$, let $\vec{x}_{<i}$ denote the induced assignment of $x_1,\cdots,x_{i-1}$. We assume that case (c.1) happens if $\Psi_{i+1}(x_1,\cdots,x_i)=1$, and case (a) happens if $\Psi_{i+1}(x_1,\cdots,x_i)=0$. Now we consider agent $i$, if $i$ is odd, agent $i$ gets a higher utility in case (c.1) than in case (a). Consequently, agent $i$ tries to set $x_i$ so that $\Psi_{i+1}(x_1,\cdots,x_i)=1$ and this is possible if and only if $(\exists x_i\Psi_{i+1}(x_1,\cdots,x_i))=1$, which is equivalent with $\Psi_{i}(x_1,\cdots,x_{i-1})=1$. If $i$ is even, agent $i$ gets a higher utility in case (a) than in case (c.1). Therefore, agent $i$ tries to set $x_i$ so that $\Psi_{i+1}(x_1,\cdots,x_i)=0$ and this is possible if and only if $(\forall x_i\Psi_{i+1}(x_1,\cdots,x_i))=0$, which is equivalent to $\Psi_{i}(x_1,\cdots,x_{i-1})=0$. To conclude, we get that case (c.1) happens if $\Psi_{i}(x_1,\cdots,x_{i-1})=1$, and (a) happens otherwise.

By induction, we know that case (c.1) happens if and only if $\Psi_{1}(\emptyset)=1$. It follows that $U_1(\vec{S}(\vec{\sigma}))\geq\frac1{n'+1}+1$ if and only if $\Psi=1$ . Therefore, we prove that the problem is a reduction from TQBF, which suggests that it is  PSPACE-hard.
\end{proof}
Even though it is PSPACE-hard to compute an SPE, we are still curious about the efficiency of the SPE. Similar to the concept of PoA, for the sequential CAG, we define the sequential price of anarchy (sPoA) (\cite{LST12}) as
$$\text{sPoA}=\sup_{\mathcal{I}}\max_{\vec{\sigma}\in \text{SPE}(\mathcal{I})}\frac{\text{SW}(\vec{S}^{*})}{\text{SW}(\vec{\alpha}(\vec{\sigma}))},$$
where $\vec{S}^{*}$ is the optimal strategy profile in term of the social welfare. In this section, we discuss the tight sPoA for two-agent setting and provide a lower bound of sPoA for general $n$-agent setting. 
\begin{theorem}\label{thm:sPoA two agents}
For any instance $\mathcal{I}$ of sequential CAG with two symmetric agents, i.e., $|A|=2$ and $w_1=w_2=1$, the sequential price of anarchy is $\frac{3}{2}$. When $|A|>2$, the sequential price of anarchy is not less than $2-1/|A|$.
\end{theorem}
\begin{proof}
Let $\vec{\sigma}$ be an SPE and $\vec{S}^{SPE}=(S_1^{SPE},S_2^{SPE})$ be the outcome of $\vec{\sigma}$. For any $\vec{S}\in\vec{\mathcal{S}}$, we prove that $\frac32 \text{SW}(\vec{S}^{SPE})\geq \text{SW}(\vec{S})$. 

First, for any $S_1\in\mathcal{S}_1$ and $S_2\in\mathcal{S}_2$, recall the potential function defined in Section \ref{Sec:basic} and observe that $U_1(S_1,S_2)=\Phi(S_1,S_2)-\Phi(\emptyset,S_2)$ and $U_2(S_1,S_2)=\Phi(S_1,S_2)-\Phi(S_1,\emptyset)$. 
For any $\vec{S}=(S_1,S_2)\in\vec{\mathcal{S}}$, let $S'_2=\sigma_2(S_1)$. Since $\vec{\sigma}$ is an SPE, by the properties of SPE, we can get the following inequalities:
\begin{align}
\Phi(S_1,S'_2)&\geq \Phi(S_1,S_2) \label{ineq:SPOA1}\\
\Phi(S_1^{SPE},S_2^{SPE})-\Phi(\emptyset,S_2^{SPE})&\geq \Phi(S_1,S'_2)-\Phi(\emptyset,S'_2) \label{ineq:SPOA2}\\
\Phi(S_1^{SPE},S_2^{SPE})&\geq \Phi(S_1^{SPE},S'_2)\label{ineq:SPOA3}
\end{align}
where inequality (\ref{ineq:SPOA1}) is due to that $U_2(S_1,S'_2)\geq U_2(S_1,S_2)$. Inequality (\ref{ineq:SPOA2}) is from $U_1(S_1^{SPE},$ $S_2^{SPE})\geq U_1(S_1,S'_2)$. Inequality (\ref{ineq:SPOA3}) holds because $U_2(S_1^{SPE},S_2^{SPE})\geq U_2(S_1^{SPE},S'_2)$.

Now observe that $$\Phi(S_1^{SPE},S'_2)-\Phi(S_1^{SPE},\emptyset)=|S'_2\setminus S_1^{SPE}|+\frac12|S_1^{SPE}\cap S'_2|\geq |S'_2|-\frac12|S_1^{SPE}|=\Phi(\emptyset,S'_2)-\frac12|S_1^{SPE}|.$$ Adding this with inequality (\ref{ineq:SPOA2}) and (\ref{ineq:SPOA3}), we get 
\begin{align*}
     2\Phi(S_1^{SPE},S_2^{SPE})-\Phi(S_1^{SPE},\emptyset)-\Phi(\emptyset,S_2^{SPE}) 
     \geq \Phi(S_1,S_2')-\frac12|S_1^{SPE}|.
\end{align*}
Then we have 
\begin{align*}
&\frac32 \text{SW}(S_1^{SPE},S_2^{SPE}) \geq U_1(S_1^{SPE},S_2^{SPE})+U_2(S_1^{SPE},S_2^{SPE})+ \frac12|S_1^{SPE}| \\
=&2\Phi(S_1^{SPE},S_2^{SPE})-\Phi(S_1^{SPE},\emptyset)-\Phi(\emptyset,S_2^{SPE})+ \frac12|S_1^{SPE}| \\
\geq&\Phi(S_1,S_2')\geq \Phi(S_1,S_2)\geq \text{SW}(S_1,S_2).
\end{align*}
Therefore, when $\vec{S}=\vec{S}^*$, we obtain sPoA$(\mathcal{I})\leq 3/2$.

To show the tightness, we provide an example whose sPoA is exactly $3/2$. Consider a sequential symmetric CAG with two agent $A=\{1,2\}$, where $N=\{q_1,q_2,q_3\}$, $s_1=\{q_1,q_2\},s_2=\{q_3\}$ and $\mathcal{S}_1=\mathcal{S}_2=\{s_1,s_2\}$. Let $\sigma_1(\emptyset)=s_1$, $\sigma_2(s_1)=s_1$, $\sigma_2(s_2)=s_1$, one can check that $\vec{\sigma}$ is an SPE, with the outcome $\vec{S}^{SPE}=(s_1,s_1)$. However, we have $\text{SW}(\vec{S}^{SPE})=2$ while the optimal social welfare equals to $\text{SW}(s_1,s_2)=3$, which means sPoA$=3/2$.

When $|A|>2$, by extending the example above, we give a lower bound of $2-\frac1{|A|}$ on sPoA, which asymptotically approaches to $2$.
\begin{example}
For $m\geq 2$, consider a sequential symmetric CAG, where $N=\{q_1,\cdots,q_{2m-1}\}$, $A=\{1,\cdots,m\}$, and for any $i \in A$, $\mathcal{S}_i=\{\{q_1,\cdots,q_m\},\{q_{m+1}\},\cdots,\{q_{2m-1}\}\}$. One can check that there is an SPE outcome where all agents selects $\{q_1,\cdots,q_m\}$, generating a social welfare of $m$. However, the optimal social welfare is $|N|=2m-1$. Therefore, the sPoA is equal to $\frac{2m-1}{m}=2-\frac1{|A|}$. 
\end{example} 
\end{proof}
\section{Conclusions and Future Work}
\label{Sec:conclusion}
In this paper, we focus on the customer attraction game in three settings. In each setting, we study the existence of pure strategy equilibrium and hardness of finding an equilibrium. Concerned about loss by competition, we also give the (s)PoA of each case.

We propose some research directions: firstly, for the asymmetric static game, we can continue exploring the sufficient conditions to guarantee the existence of PNE. Secondly, 
the result on the sPoA of the sequential game for more tha two agents is still unknown. We conjecture that the sPoA is upper-bounded by 2 when the number of agents is more than two. In addition, the weighted sequential game is also a direction for future research.


\bibliographystyle{splncs04} 
\balance
\bibliography{main}








\end{document}